\documentclass[reqno,12pt,letterpaper]{amsart}
\usepackage[proof]{zlmacros}

\title[Long time quantum--classical correspondence]{Long time quantum--classical correspondence \\ for open systems in trace norm}
\author{Zhenhao Li}
\date{}

\begin{document}

\begin{abstract}
    We consider a frictionless system coupled to an external Markovian environment. The quantum and classical evolution of such systems are described by the Lindblad and the Fokker--Planck equation respectively. We show that when such a system is given by an at most quadratically growing Hamiltonian and at most linearly growing real jump functions, the quantum and classical evolutions remain close on time scales much longer than Ehrenfest time. In particular, we show that the evolution of a density matrix by the Lindblad equation is close in trace norm to the quantization of the corresponding evolution by the Fokker--Planck equation. Such agreement improves upon recent results \cite{GZ24, HRRb, HRRa}, which proved long time agreement in weaker norms. 
\end{abstract}

\maketitle
\section{Introduction}
A statistical ensemble of quantum states is described using a density matrix, which is a positive operator of unit trace over a Hilbert space, which we will take to be $L^2(\R^n)$. It was shown by Gorini--Kossakowski--Sudarshan \cite{GKS76} in the finite dimensional case, then by Lindblad \cite{Lindblad76} in the bounded operator case, that semigroups on the Banach space of trace class operators that preserve trace and complete positivity are generated by operators of the form
\begin{equation}\label{eq:full_lindbladian}
    \mathcal L A := \frac{i}{h} [P, A] + \frac{\gamma}{h} \sum_{j = 1}^J (L_j A L_j^* - \tfrac{1}{2}(L_j^* L_j A + A L_j^* L_j)).
\end{equation}
When $L_j = 0$, this is simply the Schr\"odinger operator on density matrices, so $P$ is interpreted as the Hamiltonian. The operators $L_j$ are known as jump operators, and they describe interaction of the system with an external Markovian environment, and $\gamma$ measure the coupling strength to the external environment. The evolution equation associated with $\mathcal L$ is called the Lindblad master equation. See Chru{\'s}ci{\'n}ski--Pascazio \cite{history} for a brief survey of the equation. 

Motivated by recent works by Galkowski--Zworski with an appendix by Huang--Zworski \cite{GZ24} and Hern\'andez--Ranard--Riedel \cite{HRRb, HRRa}, we show that the quantum evolution described by the Lindblad master equation agrees with the classical evolution described by the Fokker--Planck equation on long time scales in open systems with self-adjoint jump operators. In particular, we show that the correspondence between the quantum evolution of the density matrix and the Weyl quantization of the classical evolution holds in \textit{trace norm} up to time $t \sim h^{-\ha} \gamma^{\frac{3}{2}}$ for $h^0 \le \gamma \le h^{-1}$. This improves upon \cite{GZ24}, which uses the weaker Hilbert--Schmidt norm, and on \cite{HRRb, HRRa}, which compares the quantum and classical evolution to an intermediate ansatz. Our time of correspondence is also longer than those found in previous works. However, we note that this paper does not handle friction (which require non-self-adjoint jump operators) or the $h^{\frac{1}{3}} < \gamma < h^0$ regime, which are found in \cite{GZ24, HRRb}. This type of long time correspondence is in contrast to closed systems ($\gamma = 0$), where by Egorov's theorem, the time of correspondence is $t \sim \log(1/h)$ -- see \cite[\S 11]{Zworski_semiclassical_analysis} for an overview and references, and see Figure \ref{fig:num_exp} for a numerical comparison of open and closed systems.  
\begin{figure}
    \centering
    \includegraphics[scale = 0.19]{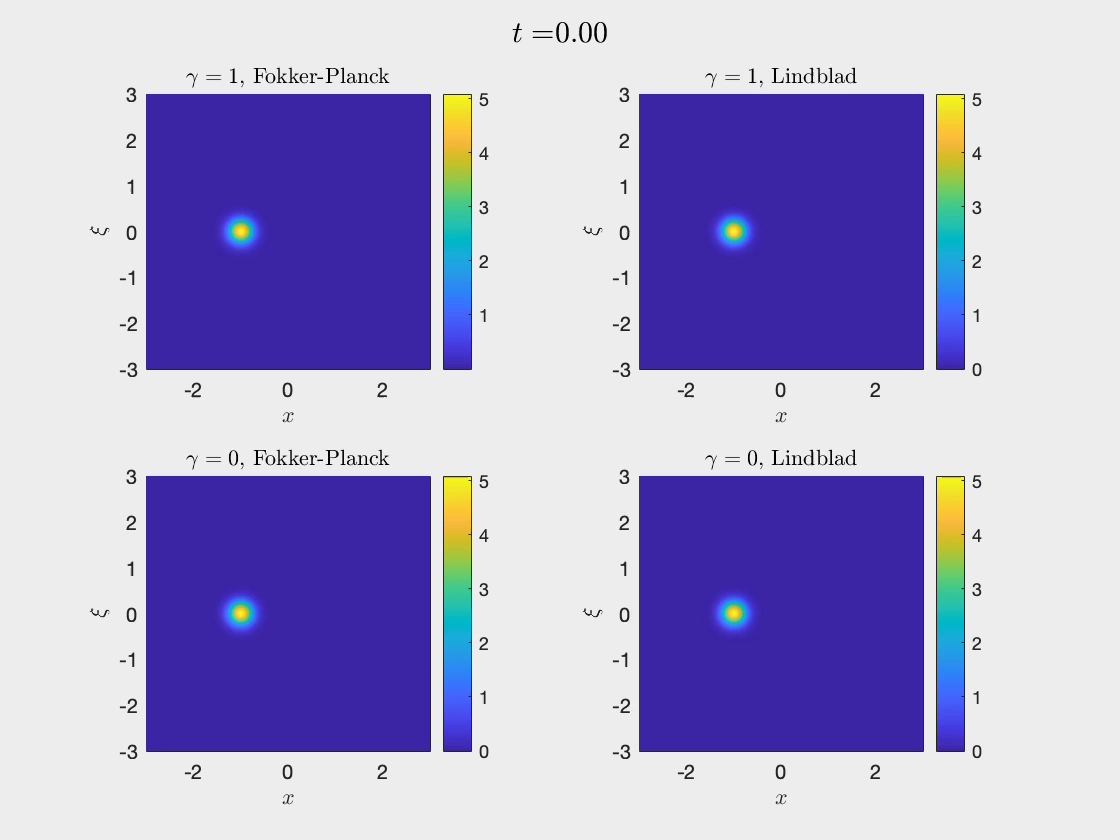}
    \includegraphics[scale = 0.19]{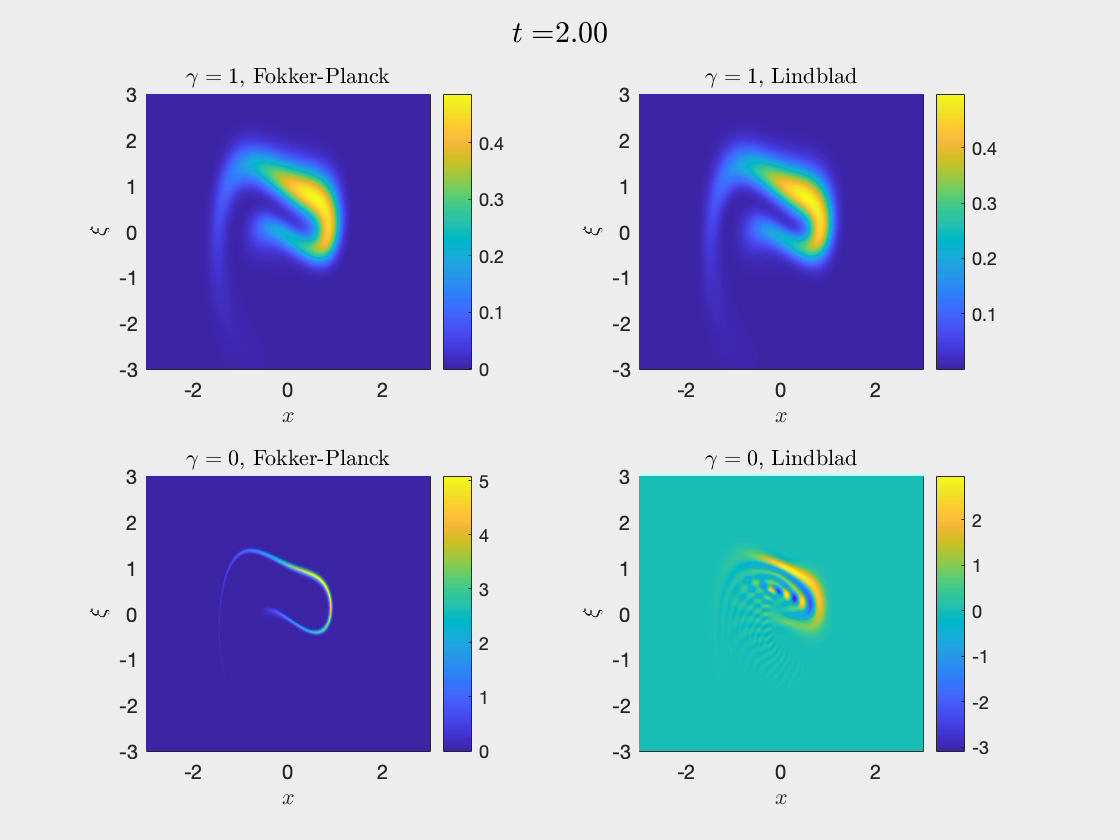}
    \caption{Numerical experiments done in the appendix of \cite{GZ24}. Coherent states (see \eqref{eq:std_coherent}) with $h = 2^{-4}$ are propagated according to the Fokker--Planck and Lindblad evolution with Hamiltonian $p(x, \xi) = \xi^2 + (x^2 - \frac{1}{2})^2$ and jump functions $\ell_1(x, \xi) = x$ and $\ell_2(x, \xi) = \xi$ (see \S\ref{sec:assumptions} for details). Depicted are contour plots of the Fokker--Planck evolution and the symbol of the density matrix of the Lindblad evolution. We see that at $t = 2$, the two evolutions are in agreement for the open system ($\gamma = 1$), but not in agreement for the closed system ($\gamma = 0$).}
    \label{fig:num_exp}
\end{figure}

\subsection{Assumptions on the Lindblad evolution}\label{sec:assumptions}
We consider $P$ and $L_j$ that are $h$-pseudodifferential operators. In particular, we assume that 
\begin{equation}\label{eq:PL_assumptions}
\begin{gathered}
    P = p^\w(x, hD), \quad |\partial^\alpha_{x, \xi} p(x, \xi)| \le C_\alpha, \quad |\alpha| \ge 2, \quad p = \bar p \\
    L_j = \ell_j^\w(x, hD), \quad |\partial^\alpha_{x, \xi} \ell_j(x, \xi)| \le C_\alpha \langle (x, \xi) \rangle^{1 - |\alpha|}, \quad |\alpha| \ge 0, \quad \ell_j = \bar \ell_j,
\end{gathered}
\end{equation}
where $\langle z \rangle := \sqrt{1 + |z|^2}$ and $j = 1, \dots, J < \infty$. Note that we require the symbols of $L_j$ to improve in decay with differentiation. The notation $a^\w(x, hD)$ is the semiclassical Weyl quantization defined in \eqref{eq:weyl}. Here, $p$ is the classical Hamiltonian and the $\ell_j$'s are called jump functions. A consequence of the Weyl quantization is that the quantization of real symbols is formally self-adjoint, so the Lindbladian defined in \eqref{eq:full_lindbladian} simplifies in our case to 
\begin{equation}\label{eq:lindbladian}
    \mathcal L A = \frac{i}{h} [P, A] - \frac{\gamma}{h} \sum_{j = 1}^J [L_j, [L_j, A]].
\end{equation}
We remark that in this case, $A(t) = \Id$ is a solution to the evolution equation $(\partial_t - \mathcal L) A(t) = 0$. Physically, this means that the fully mixed state is a steady state when the jump operators are self-adjoint, so in some sense, these are systems where the dissipation is strong. The classical counterpart to the Lindbladian is given by the leading parts of the semiclassical expansion of $\mathcal L A$. Under our assumptions~\eqref{eq:PL_assumptions}, this produces the Fokker--Planck operator
\begin{equation}
    Q:= H_p + \frac{h\gamma}{2} \sum_{j = 1}^J H_{\ell_j}^2,
\end{equation}
so the associated classical dynamics equation is given by 
\begin{equation}\label{eq:FP_def}
    (\partial_t - Q)a(x, \xi, t) = 0, \qquad a(x, \xi, 0) = a_0(x, \xi)
\end{equation}
See \S\ref{sec:param} and \S\ref{sec:L1} for details. We further make the same strong non-degeneracy assumption as in~\cite{HRRb, HRRa, GZ24}:
\begin{equation}\label{eq:jump_elliptic}
    \mathbf H \mathbf H^* \ge c I_{\R^{2n}} \quad \text{where} \quad \mathbf H:= [H_{\ell_1}, \dots ,H_{\ell_J}] \in \mathrm{Mat}_{2n \times J}(\R).
\end{equation}
This simply guarantees that $\sum_{j = 1}^J H_{\ell_j}^2$ is uniformly elliptic.

\subsection{Propagation of Gaussian states}
We state an important special case of our main result Theorem~\ref{thm:general}. We define the standard $L^2$-normalized coherent states
\begin{equation}
    \psi_{(x_0, \xi_0)} := (2 \pi h)^{-\frac{n}{4}} \exp\left( - \frac{|x - x_0|^2}{2h} + i\frac{ \langle x - x_0, \xi_0 \rangle}{h}\right).
\end{equation}
The density operator associated with $\psi_{(x_0, \xi_0)}$ is
\begin{equation}\label{eq:std_coherent}
    A_{(x_0, \xi_0)} u := \psi_{(x_0, \xi_0)} \langle u, \psi_{(x_0, \xi_0)} \rangle.
\end{equation}
Note that $\|A_{(x_0, \xi_0)}\|_{\tr} = 1$, and 
\[A_{(x_0, \xi_0)} = a^\w_{(x_0, \xi_0)}(x, hD) \quad \text{where} \quad a_{(x_0, \xi_0)}(x, \xi) = 2^n \exp \left(-\frac{|x - x_0|^2 + |\xi - \xi_0|^2}{h} \right).\]

\begin{theorem}\label{thm:gaussian_case}
    Let $\mathcal L$ be the Lindbladian defined in~\eqref{eq:full_lindbladian} and suppose that assumptions~\eqref{eq:PL_assumptions} and~\eqref{eq:jump_elliptic} hold. If $A(t)$ satisfies
    \begin{equation}
        \partial_t A(t) = \mathcal L A(t), \quad A(0) = A_{(x_0, \xi_0)},
    \end{equation}
    then for $h^0 \le \gamma \le h^{-1}$, 
    \begin{equation}
        \|A(t) - a(t)^\w(x, hD)\|_{\tr} \le \begin{cases} Ct (h^\ha + h \gamma) & 0 \le t \le 1 \\ C\left( h^\ha + h \gamma + t(h^{\ha} \gamma^{-\frac{3}{2}} + h \gamma^{-1}) \right) & t > 1 \end{cases}
    \end{equation}
    where $a(t)$ satisfies
    \begin{equation}
        (\partial_t - Q) a(t) = 0, \quad a(0) = a_{(x_0, \xi_0)}.
    \end{equation}
\end{theorem}
\begin{Remarks}
    1. We see that the time of classical--quantum correspondence for the constant to large coupling strength $h^0 \le \gamma < h^{-1}$ regime is then $t \sim h^{-\ha} \gamma^{\frac{3}{2}}$, which improves on the time of correspondence found in \cite{HRRa}, where $t \sim h^{-\ha}$. Furthermore, Theorem \ref{thm:gaussian_case} is a direct comparison of trace norms. This is in contrast to \cite{HRRb, HRRa} which compares $A(t)$ and $a(t)$ to an intermediate ansatz, and to \cite{GZ24} which compares the Hilbert--Schmidt norm. However, we stress that we only consider the case with self-adjoint jump operators in this paper, and we assume slightly stronger symbol bounds on the jump operators $L_j$ than in \cite{HRRb, GZ24}. 

    \noindent
    2. Numerical experiments were done in \cite{GZ24} that contrasts the open system case ($\gamma = 1$) to the closed system case ($\gamma = 0$). See Figure \ref{fig:num_exp}. It is clear visually that the Lindblad evolution much more closely matches the Fokker-Planck evolution in the open system case.

    \noindent
    3. The result of Theorem \ref{thm:gaussian_case} is not limited to pure Gaussian states. One can replace the initial condition $A(0)$ by a mixture of ``not-too-squeezed'' pure Gaussian states. More precisely, for $z_0= (x_0, \xi_0) \in \R^{2n}$ and $h^{-1} \sigma \in \mathrm{Sp}(2n; \R)$, define the pure Gaussian state
    \begin{equation}
        A_{z_0, \sigma} := a_{z_0, \sigma}^\w(x, hD), \quad a_{z_0, \sigma}(z) := 2^n \exp (\langle z - z_0, \sigma^{-1} (z - z_0) \rangle), \quad z = (x, \xi)
    \end{equation}
    Let $\lambda_h$ be a measure on $\mathrm{Mat}(2n \times 2n; \R) \times \R^{2n}$ such that 
    \begin{equation*}
        \supp \lambda_h \subset \left\{\sigma: \frac{\sigma}{h} \in \mathrm{Sp}(2n; \R),\, \sigma \ge h c \right\} \times \R^{2n}.
    \end{equation*}
    for some $c > 0$ independent of $h$. Then Theorem \ref{thm:gaussian_case} holds with 
    \begin{equation}
        A(0) = \int A_{z_0, \sigma} \, d \lambda_h(z_0, \sigma) \quad \text{and} \quad a(0) = \int a_{z_0, \sigma} \, d \lambda_h(z_0, \sigma).
    \end{equation}
    This is the setting of \cite{HRRb} when $\gamma \ge 1$. For Theorem~\ref{thm:gaussian_case} to hold, we just need that $A(0)$ is the quantization of an element of $S^{L^1}_{1/2}$ defined in~\eqref{eq:L1_symbol}. Such mixtures of not-too-squeezed Gaussians clearly fall into this class. We note that in the weakly coupled regime $\gamma < 1$, \cite{HRRb} actually allows for mixtures of not-too-squeezed states belonging to more exotic symbol classes. 
\end{Remarks}

\section{Symbols and quantization}\label{sec:symbols}
Classical observables can be quantized via the Weyl quantization process to obtain quantum observables. Let $\mathscr S$ denote the space of Schwartz functions. The dual $\mathscr S'$ is the space of tempered distributions. For $a \in \mathscr S'(\R^n_x \times \R^n_\xi)$ and $u \in \mathscr S(\R^n)$, 
\begin{equation}\label{eq:weyl}
    \Op^\w_h(a) u := a^\w(x, hD)u := \frac{1}{(2 \pi h)^n} \int_{\R^n} \int_{\R^n} a \Big(\frac{x + y}{2}, \xi \Big) e^{\frac{i}{h} \langle x - y, \xi \rangle} u(y)\, dy d\xi
\end{equation}
is well defined as an element of $\mathscr S'(\R^n)$ -- see \cite[Theorem 4.2]{Zworski_semiclassical_analysis}. To compose these operators, we recall the standard symbol classes. Let $m :\R^{2n} \to [0, \infty)$ be such that $m(z)/m(w) \le C \langle z - w \rangle^N$ for some $N$. The define
\begin{equation}
    a \in S_{\delta}(m) \iff |\partial_z^\alpha a(z, h)| \le C_\alpha h^{-\delta |\alpha|}m(z) \quad \text{for all} \quad z = (x, \xi) \in \R^{2n}.
\end{equation}
When $m = 1$ and $\delta = 0$, we simply write $S_{0}(1) = S$. Quantizations of such symbols are called $h$-psudodifferential operators (or $h$-pseudor for short). The point here is that for $a \in S_\delta(m)$, $\Op_h^w(a) : \mathscr S \to \mathscr S$, so the composition of $h$-pseudors makes sense, and is in fact still an $h$-pseudor -- see \cite[Theorems 4.16-18]{Zworski_semiclassical_analysis}.

The symbol classes we will use are special subsets of $S_{\delta}(m)$. First, define the symbol class
\begin{equation}
    a \in S_{(k)} \iff |\partial_{x, \xi}^\alpha a(x, \xi)|\le C_\alpha \quad \text{for all} \quad |\alpha| \ge k.
\end{equation}
Clearly, $S_{(k)} \subset S_0(\langle z \rangle^k)$. Note that from our assumptions~\eqref{eq:PL_assumptions} on the classical Hamiltonian $p$ and the jump functions $\ell_j$, we have $p \in S_{(2)}$ and $\ell_j \in S_{(1)}$. 

Next, we define the classes to which the symbol of the density matrix belongs. It is first useful to recall a characterization of the trace norm of $\Op^\w_h(a)$ in terms of the symbol $a$. 
\begin{lemma}\label{lem:trace}
    Assume that $a \in \mathscr S'(\R^{2n})$ is such that 
    \begin{equation*}
        \sum_{|\alpha| \le 2n + 1} \|\partial_{x, \xi}^\alpha a\|_{L^1} < \infty.
    \end{equation*}
    Then $\Op^\w(a)$ is in trace class and
    \begin{equation}
        \|\Op^\w_h(a)\|_{\tr} \le C h^{-n} \sum_{|\alpha| \le 2n + 1} h^{\frac{\alpha}{2}} \|\partial^{\alpha}_{x, \xi} a\|_{L^1}.
    \end{equation}
\end{lemma}
\begin{proof}
We quickly show how this is obtained by rescaling the non-semiclassical trace norm estimate. Observe that 
\begin{equation}
    \Op_h^\w(a) = (M^{-1})^*\Op^\w(\tilde a) M^*
\end{equation}
where $M^*$ is the pullback by $M(x) = h^\ha x$, $\tilde a(x, \xi) = a(h^\ha x, h^\ha \xi)$, and
\begin{equation*}
    \Op^\w(\tilde a) u := \frac{1}{(2 \pi)^n} \int_{\R^n} \int_{\R^n} \tilde a\Big( \frac{x + y}{2}, \xi \Big) e^{i \langle x - y, \xi \rangle} u(y)\, dy d\xi.
\end{equation*}
It follows from \cite[Theorem 9.3]{DS99} that 
\begin{equation}
    \|\Op^\w(\tilde a)\|_{\tr} \le C \sum_{|\alpha| \le 2n + 1} \|\partial^\alpha_{x, \xi} \tilde a\|_{L^1},
\end{equation}
so it follows from \cite[(C.3.1)]{Zworski_semiclassical_analysis}
\[\|\Op^\w_h(a)\|_{\tr} = \|\Op^\w(\tilde a)\|_{\tr} \lesssim h^{-n} \sum_{|\alpha| \le 2n + 1} h^{\frac{\alpha}{2}} \|\partial^{\alpha} a\|_{L^1}\]
as desired.
\end{proof}

The symbol of the density matrices we are interested in lie in the $L^1$-based symbol class 
\begin{equation}\label{eq:L1_symbol}
    a \in S^{L^1}_\rho \iff h^{-n} \|\partial^\alpha a\|_{L^1} \le C_\alpha h^{-\rho |\alpha|}
\end{equation}
An important special case here is when $\rho = \ha$. Indeed, we see that the density matrices associated with standard coherent states defined in~\eqref{eq:std_coherent} are quantizations of symbols in the class $S^{L^1}_{1/2}$. Note that the coherent states have unit trace, so Lemma~\ref{lem:trace} gives optimal dependence on $h$ for $S^{L^1}_{1/2}$ symbols.

By the Sobolev embedding $W^{n, 1}(\R^n) \hookrightarrow L^\infty(\R^n)$, we see that
\begin{equation}\label{eq:basic_S_embedding}
    S^{L^1}_\rho \subset h^{n(1 - 2\rho)} S_\rho(1)
\end{equation}
We need to compose $S^{L^1}_\rho$ symbols with symbols $S_{(k)}$ and understand the asymptotic expansion of the composition. Define the tensor product 
\begin{equation}
    c(z, w) \in S \otimes S^{L^1}_\rho \iff h^{-n} \|\sup_{z} |\partial_z^\alpha \partial_w^\beta c(z, \bullet)|\|_{L^1} \le C_{\alpha \beta} h^{-\rho|\beta|}, \quad z, w \in \R^{2n}
\end{equation}
We first have the following lemma, which is the $L^1$ counterpart to \cite[Lemma 2.1]{GZ24}.
\begin{lemma}
    Let $Q: \R^{2n} \times \R^{2n}$ be a non-degenerate bilinear quadratic form. Then 
    \begin{equation}
        e^{ihQ(D_z, D_w)}: S \otimes S^{L^1}_\rho \to S \otimes S^{L^1}_\rho
    \end{equation}
    is continuous. Furthermore, we have the asymptotic expansion
    \begin{equation}
        e^{ihQ(D_z, D_w)} a - e^{\frac{i \pi}{4} \sgn Q} \sum_{k = 0}^{N - 1} \Big( \frac{h}{i} \Big)^k \frac{1}{k!} Q(D_z, D_w)^k a(z, w) \in h^{N(1 - \rho)} S \otimes S^{L^1}_\rho
    \end{equation}
    for every $N$. 
\end{lemma}
\begin{proof}
Let
\begin{equation}\label{eq:eQa}
    c(z, w):= e^{ihQ(D_z, D_w)} a(z, w)
\end{equation}
where $Q(\zeta, \omega):= \ha \langle B(\zeta, \omega), (\zeta, \omega)\rangle$ is a non-degenerate bilinear quadratic form and $a \in S \otimes S^{L^1}_{\rho}$. By \eqref{eq:basic_S_embedding}, we see that \eqref{eq:eQa} is indeed well-defined and 
\begin{equation}
    c(z, w) = \frac{|\det B|^{-\ha}}{(2 \pi h)^{2n}} \iint e^{\frac{i}{h} \varphi(z_1, w_1)} a(z + z_1, w + w_1) \, dz_1 dw_1
\end{equation}
can be understood as an oscillatory integral, where we have the quadratic phase given by
\begin{equation}
    \varphi(z_1, w_1):= - \ha \langle B^{-1}(z_1, w_1), (z_1, w_1) \rangle.
\end{equation}
We change variables to $v_1 = h^{-\rho} w_1$, and we apply a cutoff $\chi \in \CIc(\R^{4n})$ such that $\chi = 1$ near $0$ and $\chi(z_1, v_1) = 0$ when $|(z_1, v_1)| \ge 1$. Then $c$ is given by 
\begin{align*}
    c(z, w) &= \frac{|\det B|^{\ha}}{(2 \pi h^{1 - \rho})^{2n}} \iint e^{\frac{i}{h^{1 - \rho}} \varphi(z_1, v_1)} a(z + z_1, w + h^\rho v_1) \, dz_1 dv_1 \\
    &= \frac{|\det B|^{\ha}}{(2 \pi h^{1 - \rho})^{2n}} \iint e^{\frac{i}{h^{1 - \rho}} \varphi(z_1, v_1)} \chi(z_1, v_1) a(z + z_1, w + h^\rho v_1) \, dz_1 dv_1 \\
    &\qquad \quad + \frac{|\det B|^{\ha}}{(2 \pi h^{1 - \rho})^{2n}} \iint e^{\frac{i}{h^{1 - \rho}} \varphi(z_1, v_1)} (1 - \chi(z_1, v_1)) a(z + z_1, w + h^\rho v_1) \, dz_1 dv_1 \\
    &=: c_1(z, w) + c_2(z, w).
\end{align*}
We first analyze $c_1$ by stationary phase, which gives the formal expansion 
\begin{equation}\label{eq:stat_phase}
    c_1(z, w) \sim e^{\frac{i \pi}{4} \sgn B} \sum_{k = 0}^\infty \Big(\frac{h^{1 - \rho}}{i}\Big)^k \frac{1}{k!} Q(D_{z_1}, D_{v_1})^k a(z + z_1, w + h^\rho v_1)|_{z_1 = v_1 = 0}.
\end{equation}
Let $c_{1, N}(z, w)$ be the $N$-term truncation of the formal series~\eqref{eq:stat_phase}. The error is given by 
\begin{align*}
    &|\partial_z^{\alpha_1} \partial_w^{\alpha_2} (c_1(z, w) - c_{1, N}(z, w))| \\
    &\le C_N h^{(1 - \rho)N} \sum_{|\beta_1| + |\beta_2| \le 2 N + 4n + 1} h^{-\rho|\alpha_2|} \sup_{|(z_1, v_1)| \le 1} |\partial_{z_1}^{\alpha_1 + \beta_1} \partial_{v_1}^{\alpha_2 + \beta_2} a(z + z_1, w + h^\rho v_1)| \\
    &=: C_N h^{(1 - \rho) N} \sum_{|\beta_1| + |\beta_2| \le 2 N + 4n + 1} R_{\alpha, \beta}(z, w).
\end{align*}
By Sobolev embedding, we see that 
\begin{equation}
    |R_{\alpha, \beta}(z, w)| \le h^{-\rho|\alpha_2|} \sum_{|\gamma| \le 2n} \|\partial_{(z_1, v_1)}^{\alpha + \beta + \gamma} a(z + \bullet, w + h^\rho \bullet)\|_{L^1(B_{\R^{4n}}(0, 1))}
\end{equation}
Then
\begin{align*}
    &\|\sup_z |R_{\alpha, \beta}(z, \bullet)|\|_{L^1(\R^{2n})} \\
    &\le h^{-\rho|\alpha_2|} \int_{\R^{2n}} \int_{|(z_1, v_1)|\le 1} \sup_z |\partial_{(z_1, v_1)}^{\alpha + \beta + \gamma} a(z + z_1, w + h^\rho v_1)|\, dz_1 dv_1dw \\
    &\le h^{\rho(|\beta_2| + |\gamma_2|)} \int_{B_{\R^{4n}}(0, 1)}\|\sup_z|\partial_{(z, w)}^{\alpha + \beta + \gamma} a(z + z_1, w)|\|_{L^1_w}\, dz_1 dv_1 \\
    &\le h^{n - \rho|\alpha_2|}
\end{align*}
In particular, this means that 
\[c_1(z, w) - c_{1, N}(z, w)\in h^{(1 - \rho)N} S \otimes S^{L^1}_\rho.\]
Now we estimate $c_2$, which we should expect to be a residual term. We do this by integration by parts. We see that for $N > 2n + 1$, uniformly in $z$, we have the estimate
\begin{align*}
    &h^{2n(1 - \rho)}\|\sup_z |\partial_{(z, w)}^\alpha c_2(z, \bullet)|\|_{L^1(\R^{2n})} \\
    &\le C\int \sup_z \left| \iint e^{\frac{i}{h^{1 - \rho}} \varphi(z_1, v_1)}(1 - \chi(z_1, v_1)) \partial_{(z, w)}^\alpha a(z + z_1, w + h^\rho v_1)\, dz_1 dv_1 \right| dw \\
    &\le Ch^{N(1 - \rho)} \iiint \langle(z_1, v_1) \rangle^{-N} \sum_{|\beta| \le N} \sup_z |\partial^\alpha_{(z, w)} \partial_z^{\beta_1} (h^\rho \partial_w)^{\beta_2} a(z + z_1, w + h^\rho v_1)| \, dz_1 dv_1 dw \\
    &\le C h^{N(1 - \rho)} \iint \langle (z_1, v_1) \rangle^{-N} \sum_{|\beta| \le N} \|\sup_z|\partial^\alpha_{(z, w)} \partial_z^{\beta_1} (h^\rho \partial_w)^{\beta_2} a(z + z_1, w)|\|_{L^1_w}\, dz_1\, dv_1 \\
    &\le C h^{N(1 - \rho) + n - \rho|\alpha_2|}.
\end{align*}
Therefore, $c_2 \in h^{(N - 2n)(1 - \rho)} S \otimes S^{L^1}_\rho$ for arbitrary $N$. Hence it follows that 
\begin{equation*}
    c(z, w) - c_{1, N}(z, w) \in h^{(1 - \rho) N} S \otimes S^{L^1}_\rho
\end{equation*}
as desired.
\end{proof}

The lemma above gives us the desired composition properties of $S_{(k)}$ with $S^{L^1}_\rho$. 
\begin{lemma}\label{lem:composition}
    Let $a \in S_{(k)}$ and $b \in S^{L^1}_\rho$ for $k \in \N$ and $0 \le \rho < 1$. Then there exists $c \in \mathscr S'$ such that 
    \[\Op_h^\w(a) \Op_h^\w(b) = \Op_h^\w(c),\]
    and
    \begin{equation}
        c(x, \xi) - \sum_{j = 0}^{N - 1} \frac{1}{j!} \left( \frac{h}{2 i} \sigma(D_x, D_\xi, D_y, D_\eta)^j a(x, \xi) b(y, \eta) \right)\bigg |_{\substack{y = x \\ \eta = \xi}} \in h^{N(1 - \rho)} S_\rho^{L^1_\rho}
    \end{equation}
    for all $N \ge k$, where $\sigma(x, \xi, y, \eta) := \langle \xi, y \rangle - \langle x, \eta \rangle$ denotes the standard symplectic form.
\end{lemma} 
It will also be clear from the proof that the analogous result holds for $\Op_h^\w(b) \Op_h^\w(a)$.
\begin{proof}
Put $z = (x, \xi)$ and $w = (y, \eta)$. We have
\begin{equation}
    \Op_h^\w(a) \Op_h^\w(b) = \Op_h^\w(c), \qquad c(z) := e^{ih A(D_{z, w})} a(z) b(w)|_{z = w}
\end{equation}
where $A(D_{z, w}) := \sigma(D_x, D_\xi, D_y, D_\eta)$. Then $c$ has the asymptotic development 
\begin{equation}
    c(z) = \sum_{\ell = 0}^{N - 1} \frac{1}{\ell!}(ihA(D))^\ell (a(z) b(w))|_{z = w} + r_N(z)
\end{equation}
where
\begin{equation}\label{eq:r_N}
    r_N(z) := \frac{1}{(N - 1)!} \int_0^1 (1 - t)^{N - 1} e^{ithA(D)}(ihA(D))^N (a(z) b(w))|_{z = w} \, dt.
\end{equation}
Note that for $N \ge k$,
\[A(D)^N (a(z) b(w)) \in h^{-N \rho} S \otimes S^{L^1}_\rho.\]
We also have that $e^{ihtA(D)}: S \otimes S^{L^1}_\rho \to S \otimes S^{L^1}_\rho$ is uniformly bounded for $t \in (0, 1)$. For any $e \in S \otimes S^{L^1}_\rho$, we have
\begin{equation}
    \|\partial_w^\alpha e(w, w)|\|_{L^1} \le \sum_{|\beta| \le |\alpha|} \|\sup_z|\partial_{(z, w)}^\beta e(z, \bullet)|\|_{L^1}.
\end{equation}
Therefore, we see from~\eqref{eq:r_N} that indeed $r_N \in h^{N(1 - \rho)} S^{L^1}_\rho$. 
\end{proof}

\section{Fokker--Planck parametrix in the small diffusion limit}\label{sec:param}
The classical dynamics is described by the Fokker--Planck equation~\eqref{eq:FP_def}. Under our assumptions~\eqref{eq:PL_assumptions} and~\eqref{eq:jump_elliptic} for $p$ and $\ell_j$, the Fokker--Planck equation is a second-order parabolic equation on $\R^{2n}$ whose diffusion coefficient tends to zero. In the following two sections, we independently study the $L^1$ properties of such equations. We consider the $\epsilon$-dependent operator
\[Q = \epsilon^2 \nabla \cdot A(x) \nabla + v(x) \cdot \nabla, \qquad x \in \R^n,\]
where $A \in C^\infty(\R^n; \mathrm{Sym}_{n \times n}(\R))$ and $v \in C^\infty(\R^n; \R^n)$ satisfy the following conditions:
\begin{gather}
    c \le A \le c^{-1} \label{eq:uniform_ell_est}\\
    \partial_x^\alpha A_{jk}(x) \le C_{\alpha, j, k} \langle x \rangle^{-|\alpha|} \quad \text{for all} \quad \alpha \label{eq:A_est}\\
    \partial_x^\alpha v(x) \le C_{\alpha} \quad \text{for all} \quad |\alpha| \ge 1 \label{eq:v_est} \\
    \nabla \cdot v = 0 \label{eq:divfree}
\end{gather}
Most importantly, we note that the Fokker--Planck operator we introduced in~\eqref{eq:FP_def} under the assumptions~\eqref{eq:PL_assumptions} and~\eqref{eq:jump_elliptic} are of this form for $\epsilon = \sqrt{\gamma h/2}$. The first condition~\eqref{eq:uniform_ell_est} gives us uniform ellipticity of the principal term. Conditions~\eqref{eq:A_est} and~\eqref{eq:v_est} essentially say that changes in $A$ occur on the same scale as the dynamics given by the vector field $v$; the dynamics is faster near infinity, and we need $A(x)$ to be comparable to $A(\varphi^{-1}(x))$ where $\varphi^t$ is the flow map of $v$. Finally, condition \eqref{eq:divfree} guarantees conservation of mass. For the parametrix construction in this section, condition~\eqref{eq:divfree} is not needed, but it is important for the following section in obtaining $L^1$ estimates. 

We consider the Cauchy problem
\begin{equation}\label{eq:cauchy_problem}
    \begin{cases}
        (\partial_t - Q)u = 0\\
        u(x, 0) = u_0(x)
    \end{cases}
\end{equation}
Our goal is to have long time control of the $L^1$-norms of spatial derivatives of the solution in terms of the initial data. This control is established Proposition \ref{prop:short_time}, which gives Sobolev estimates up to constant time, and in Proposition \ref{prop:smoothing}, which gives $\epsilon$-semiclassical smoothing estimates past constant time. To establish either of these estimates, we need a good parametrix.

The idea for our parametrix construction is based on the standard parabolic parametrix constructed in~\cite{Levi07}. The key difference is that in the small diffusion limit as $\epsilon \to 0$, the dynamics dominates over the diffusion. In order to obtain a good parametrix uniformly in $\epsilon$, the parametrix must follow the dynamics. 

\subsection{First approximation} Let $\varphi^t$ denote the flow generated by the vector field $v(x)$. Define
\begin{equation}
    K_0(x, y, t) := c_n (\epsilon^2 t)^{-n/2} \det A(x)^{-\ha} \exp \left(-\frac{\langle x - \varphi^{-t}(y), A^{-1}(x)(x - \varphi^{-t}(y))\rangle}{4\epsilon^2 t} \right)
\end{equation}
where $c_n := (4 \pi)^{-n/2}$. This choice of normalization ensures that 
\begin{equation*}
    K_0(x, y, t) \to \delta_0(x - y) \quad \text{in distributions as} \quad t \to 0^+. 
\end{equation*}
Define
\begin{equation}\label{eq:R1_def}
    R_1(x, y, t) := -(\partial_t - Q) K_0(x, y, t).
\end{equation}
We have the following estimate on $R_1$. 
\begin{lemma}
For $\epsilon, t \in (0, 1]$, $R_1(x, y, t)$ satisfies the pointwise bound 
\begin{equation}\label{eq:R1_pointwise}
    |R_1(x, y, t)| \le C (\epsilon^2 t)^{-\frac{n}{2}}(1 + \epsilon t^{-\ha}) \exp \left(- c \frac{|x - \varphi^{-t}(y)|^2}{\epsilon^2 t} \right)
\end{equation}
where $C, c > 0$ are independent of $\epsilon > 0$. This also implies the $L^1$ estimate
\begin{equation}
    \|R_1(\bullet, y, t)\|_{L^1_x} \le C(1 + \epsilon t^{-\ha}).
\end{equation}    
\end{lemma}
Note that this estimate is uniformly $\mathcal O(t^{-\ha})$ in $\epsilon$ for $\epsilon, t \in (0, 1]$.
\begin{proof}
The proof follows more or less by a direct computation, which we must carry out in some detail. The time derivative of $K_0$ given by 
\begin{align}
    \frac{\partial_t K_0(x, y, t)}{K_0(x, y, t)} &= -\frac{n}{2} t^{-1} - \partial_t \left( \frac{ \langle x - \varphi^{-t}(y), A^{-1}(x)(x - \varphi^{-t}(y)) \rangle}{4\epsilon^2 t} \right) \nonumber \\
    &= - \frac{n}{2} t^{-1} + \frac{ \langle x - \varphi^{-t}(y), A^{-1}(x)(x - \varphi^{-t}(y)) \rangle}{4\epsilon^2 t^2} \nonumber \\
    &\quad \quad - \frac{\langle v(\varphi^{-t}(y)), A^{-1}(x)(x - \varphi^{-t}(y)) \rangle}{2\epsilon^2 t}. \label{eq:time}
\end{align}
The transport term is given by 
\begin{multline}\label{eq:transport}
    \frac{v(x) \cdot \nabla K_0(x, y, t)}{K_0(x, y, t)} = - \frac{1}{2}\frac{v(x) \cdot \nabla \det A(x)}{\det A(x)}  - \frac{\langle v(x), A^{-1}(x)(x - \varphi^{-t}(y)) \rangle}{2\epsilon^2 t} \\
    - \frac{\langle x - \varphi^{-t}(y), (v(x) \cdot \nabla A^{-1}(x))(x - \varphi^{-t}(y)) \rangle}{4\epsilon^2 t}.
\end{multline}
Observe that by~\eqref{eq:uniform_ell_est}, \eqref{eq:A_est}, and~\eqref{eq:v_est}, the first term on the right-hand-side of~\eqref{eq:transport} has uniformly bounded derivatives:
\begin{equation*}
    \partial_x^\alpha \left(\frac{1}{2}\frac{v(x) \cdot \nabla \det A(x)}{\det A(x)} \right) \le C_\alpha
\end{equation*}
for all $\alpha \in \N_0^n$. 
Finally, we consider the diffusive term. First, observe that 
\begin{multline*}
    \frac{ A(x) \nabla K_0(x, y, t)}{K_0(x, y, t)} =\det A(x)^{\ha} A(x) \nabla(\det A(x)^{-\ha}) \\
    - \frac{x - \varphi^{-t}(y)}{2\epsilon^2 t} - \frac{\langle x - \varphi^{-t}(y), \nabla A^{-1}(x)(x - \varphi^{-t}(y)) \rangle}{4\epsilon^2 t}.
\end{multline*}
It follows from~\eqref{eq:uniform_ell_est} and~\eqref{eq:A_est} that
\begin{equation*}
    \left|\partial_x^\alpha \left( \det A(x)^{\ha} A(x) \nabla(\det A(x)^{-\ha}) \right) \right| \le C_\alpha
\end{equation*}
for all $\alpha \in \N_0^n$. Then taking the divergence, we see that the diffusion term takes the form
\begin{align}
    \frac{\epsilon^2 \nabla \cdot A(x)\nabla K_0(x, y, t)}{K_0(x, y, t)} &= -\frac{n}{2} t^{-1} + \frac{\langle x - \varphi^{-t}(y), A^{-1}(x)(x - \varphi^{-t}(y)) \rangle}{4\epsilon^2 t^2} + \epsilon^2 \mathcal P_0(x)\nonumber \\
    &\qquad +\frac{ \mathcal P_1(x, x - \varphi^{-t}(y))}{t} + \frac{\mathcal P_2(x, x - \varphi^{-t}(y))}{t} \nonumber \\
    &\qquad \qquad + \frac{\mathcal P_3(x, x - \varphi^{-t}(y))}{\epsilon^2 t^2} + \frac{\mathcal P_4(x, x - \varphi^{-t}(y))}{\epsilon^2 t^2} \label{eq:diffusion}
\end{align}
where $\mathcal P_j(x, \eta) \in S_{\hom}^j(\R^n \times \R^n)$ is a homogeneous symbol of order $j$, that is, 
\begin{equation}
    \mathcal P_j(x, \eta) \in S_{\hom}^j(\R^n \times \R^n) \iff P_j(x, \eta) = \sum_{|\alpha| = j} a_\alpha(x) \eta^\alpha, \quad |\partial_x^\beta a_\alpha| \le C_{\alpha, \beta}
\end{equation}
The fact that $\mathcal P_j$ satisfies symbol estimates follows from uniform ellipticity~\eqref{eq:uniform_ell_est} and the symbol estimates \eqref{eq:A_est}. Although it is not the case for $\mathcal P_j$ here, we will also allow the coefficients $a_\alpha$ of members of $S^j_{\hom}$ to depend on $\epsilon$ but satisfies symbol estimates uniformly in $\epsilon$. It is convenient to record the property 
\begin{equation}\label{eq:Pj_derivatives}
    \begin{gathered}
        \partial_x \mathcal P_j(x, x - \varphi^{-t}(y)) = \mathcal Q_j(x, x - \varphi^{-t}(y)) + \mathcal Q_{j - 1}(x, x - \varphi^{-t}(y)) \\
        \partial_x \mathcal P_j(x, y) = \mathcal Q_j(x, y) \\
        \partial_y \mathcal P_j(x + \varphi^{-t}(y), x) = \mathcal Q_j(x + \varphi^{-t}(y), x)
    \end{gathered}
\end{equation}
where $\mathcal Q_j \in S^j_{\hom}(\R^n \times \R^n)$ and may change from line to line. 

Combining~\eqref{eq:time}, \eqref{eq:transport}, and~\eqref{eq:diffusion}, we see that
\begin{equation}\label{eq:R_1_expansion}
\begin{aligned}
    R_1(x, y, t) &= \Bigg[\frac{\langle v(x) - v(\varphi^{-t}(y)), A^{-1}(x)(x - \varphi^{-t}(y)) \rangle }{2\epsilon^2 t}  \\
    & \quad + \mathcal P_0(x) + \frac{ \mathcal P_1(x - \varphi^{-t}(y), x)}{t} + \frac{\mathcal P_2(x - \varphi^{-t}(y), x)}{\epsilon^2 t} \\
    & \quad \quad+ \frac{\mathcal P_3(x - \varphi^{-t}(y), x)}{\epsilon^2 t^2} + \frac{\mathcal P_4(x - \varphi^{-t}(y), x)}{\epsilon^2 t^2}\Bigg] K_0(x, y, t)
\end{aligned}
\end{equation}
where $\mathcal P_j \in S^j_{\hom}(\R^n \times \R^n)$ (and are possibly different from the $\mathcal P_j$ from~\eqref{eq:diffusion}). It follows from \eqref{eq:v_est} that for all $t \in [0, 1]$, 
\begin{equation}\label{eq:no_warping}
    |v(x) - v(\varphi^{-t}(y))| \le C|x - \varphi^{-t}(y)|.
\end{equation}
The pointwise estimate then follows upon applying~\eqref{eq:no_warping} to the first term on the right-hand-side of~\eqref{eq:R_1_expansion}.
\end{proof}
We emphasize that the pointwise estimate~\eqref{eq:R1_pointwise} holds uniformly in $\epsilon$ precisely because $K_0$ is roughly a Gaussian that follows the dynamics generated by the vector field $v(x)$. The standard parabolic parametrix construction does not follow the dynamics generated by the sub-leading order term, and thus one cannot expect to obtain uniform estimates in $\epsilon$ from that construction.

\subsection{Higher order corrections} 
Now we proceed to make corrections to $K_0$ to improve the parametrix. The idea is the same as in the standard parabolic parametrix construction. Assume that the fundamental solution to \eqref{eq:cauchy_problem} is of the form 
\begin{equation*}
    K(x, y, t) = K_0(x, y, t) + \int_0^t \int_{\R^n} K_0(x, z, t - s) R(z, y, s)\, ds.
\end{equation*}
Then we see that 
\begin{equation*}
    R(x, y, t) = R_1(x, y, t) + \int_0^t \int_{\R^n} R_1(x, z, t - s) R(z, y, s)\, ds.
\end{equation*}
Define
\begin{equation}\label{eq:Rk_def}
    R_k := \int_0^t \int_{\R^n} R_1(x, z, t - s) R_{k - 1}(z, y, s)\, ds.
\end{equation}
Then at least formally, we have that
\begin{equation*}
    R = \sum_{k = 1}^\infty R_k
\end{equation*}
A good candidate for an improved parametrix is then given by 
\begin{equation}\label{eq:Kj_def}
    K_j := K_0(x, y, t) + \sum_{k = 1}^j \int_0^t \int_{\R^n} K_0(x, z, t - s) R_j(z, y, s)\, ds
\end{equation}
In particular, we see that 
\begin{equation}
    (\partial_t - Q) K_j(x, y, t) = -R_1 + \sum_{k = 1}^j(R_k - R_{k + 1})= -R_{j + 1}
\end{equation}
To justify that this is a parametrix, we must estimate $R_j$ as well as its derivatives. We first need derivative estimates on $R_1$. 
\begin{lemma}\label{lem:R1_est}
Let $R_1(x, y, t)$ be as defined in \eqref{eq:R1_def}. For $\epsilon, t \in (0, 1]$, $R_1$ satisfies the pointwise estimates
\begin{equation}\label{eq:dxR1}
    |\partial_x^\alpha R_1(x, y, t)| \le C_\alpha (\epsilon^2 t)^{-\frac{n + |\alpha|}{2}} t^{-\ha} \exp \left(- c_\alpha \frac{|x - \varphi^{-t}(y)|^2}{\epsilon^2 t} \right),
\end{equation}
\begin{equation}\label{eq:dxR1'}
    |\partial_x^\alpha R_1(x + y, \varphi^{t}(x), t)| \le C_\alpha(\epsilon^2 t)^{-\frac{n}{2}} t^{-\ha} \exp \left(-c_\alpha \frac{|y|^2}{\epsilon^2 t}\right),
\end{equation}
where $C_\alpha, c_\alpha > 0$ are independent of $\epsilon$ and $t$.   
\end{lemma}
We see from \eqref{eq:dxR1} that higher derivatives in $x$ leads to faster rate of blow up as $\epsilon^2 t \to 0$. On the other hand, the purpose of \eqref{eq:dxR1'} is to show that trading $x$ derivatives of $R_1$ for $y$ derivatives is not too costly, since \eqref{eq:dxR1'} does not become worse with more derivatives. We also remark that it will also be useful to rearrange \eqref{eq:dxR1'} as 
\begin{equation}\label{eq:dyR1'}
    |\partial_y^\alpha R_1(x + \varphi^{t}(y), y, t)| \le C_\alpha(\epsilon^2 t)^{-\frac{n}{2}} t^{-\ha} \exp \left(-c_\alpha \frac{|x|^2}{\epsilon^2 t}\right).
\end{equation}
\begin{proof}
1. We take derivatives of the expansion for $R_1$ given in~\eqref{eq:R_1_expansion}. We first consider derivatives hitting $K_0$. By conditions~\eqref{eq:uniform_ell_est} and~\eqref{eq:A_est}, for any $\alpha \in \N_0^n$, there exists $\mathcal P_j \in S^j_\hom(\R^n \times \R^n)$ such that 
\begin{equation}
    \partial^\alpha_x K_0(x, y, t) = \sum_{j = 0}^{|\alpha|} \frac{\mathcal P_j(x, x - \varphi^{-t}(y))}{(\epsilon^2 t)^{\frac{|\alpha| + j}{2}}} K_0(x, y, t).
\end{equation}
Here, $\mathcal P_j$ may depend on $\epsilon, t \in (0, 1]$ but lie in the symbol class $S^j_\hom(\R^n \times \R^n)$ uniformly. Then it follows that 
\begin{equation}\label{eq:K0_derivatives}
    |\partial_x^\alpha K_0(x, y, t)| \le C_\alpha (\epsilon^2 t)^{-\frac{n + |\alpha|}{2}} \exp\left(-c_\alpha \frac{|x - \varphi^{-t}(y)|^2}{\epsilon^2 t}\right)
\end{equation}
for some $c_\alpha > 0$. 

We also need to control derivatives of the first term of on the right-hand-side of \eqref{eq:R_1_expansion}. For this, we compute
\begin{align*}
    &\partial_{x_\ell} \langle v(x) - v(\varphi^{-t}(y)), A^{-1}(x)(x - \varphi^{-t}(y)) \rangle \\
    =& \langle v(x) - v(\varphi^{-t}(y)), (\partial_{x_\ell}A^{-1})(x)(x - \varphi^{-t}(y)) \rangle \\
    &\quad + \langle (\partial_{x_\ell} v)(x), A^{-1}(x)(x - \varphi^{-t}(y)) \rangle + \langle v(x) - v(\varphi^{-t}(y)), A^{-1}(x)\delta_\ell \rangle
\end{align*}
where $\delta_\ell$ is the vector with a $1$ in the $\ell$-th component and zeros elsewhere. Higher derivatives take a similar form, and it is easy to see that 
\begin{equation}\label{eq:quadform_derivatives1}
    \partial_x^\alpha \langle v(x) - v(\varphi^{-t}(y)), A^{-1}(x)(x - \varphi^{-t}(y)) \rangle \le C_\alpha(|x_k - \varphi^{-t}(y)|^2 + |x_k - \varphi^{-t}(y)|)
\end{equation}
for $|\alpha| = 1$
and
\begin{equation}\label{eq:quadform_derivatives2}
    \partial_x^\alpha \langle v(x) - v(\varphi^{-t}(y)), A^{-1}(x)(x - \varphi^{-t}(y)) \rangle \le C_\alpha(|x_k - \varphi^{-t}(y)|^2 + |x_k - \varphi^{-t}(y)| + 1)
\end{equation}
for $|\alpha| \ge 2$. 
The derivative estimates \eqref{eq:K0_derivatives}, \eqref{eq:quadform_derivatives1} and \eqref{eq:quadform_derivatives2} combined with the symbol derivative estimates \eqref{eq:Pj_derivatives} yields the estimate \eqref{eq:dxR1}. 

\noindent
2. Now we estimate $R_1(x + y, \varphi^t(x), t)$. Similarly, we first consider $K_0(x + y, \varphi^t(x), t)$. It is easy to see that 
\begin{equation}
    \partial_x^\alpha K_0(x + y, \varphi^{t}(x), t) \le C_\alpha (\epsilon^2 t)^{-\frac{n}{2}} \exp\left(- c_\alpha \frac{|y|^2}{\epsilon^2 t} \right).
\end{equation}
Next, we compute
\begin{equation}
    \partial_{x_\ell}^\alpha \langle v(x + y) - v(x), A^{-1}(x + y) y \rangle
    = \sum_{\beta = 0}^\alpha C_{\alpha, \beta} \langle \partial_x^\beta(v(x + y) - v(x)), \partial_x^{\alpha - \beta} A^{-1}(x + y) y \rangle,
\end{equation}
which yields 
\begin{equation}
    |\partial_{x_\ell}^\alpha \langle v(x + y) - v(x), A^{-1}(x + y) y \rangle| \le C_\alpha |y|^2.
\end{equation}
Combining with the symbol derivative estimates~\eqref{eq:Pj_derivatives} yields~\eqref{eq:dxR1'}.
\end{proof}

Now we estimate $R_k$ and its derivatives, which in turn allow us to estimate derivatives of $K_j$.
\begin{lemma}\label{lem:Rk_est}
    Let $R_k$ be the correction terms defined in~\eqref{eq:Rk_def} and $K_j$ be the parametrix defined in~\eqref{eq:Kj_def}. Then for $\epsilon, t \in (0, 1]$, $R_k$ satisfy the pointwise estimate 
    \begin{equation}\label{eq:Rk_est}
        |\partial_x^\alpha R_k(x, y, t)| \le C_{k, \alpha} (\epsilon^2 t)^{-\frac{n + |\alpha|}{2}} t^{\frac{k}{2} - 1} \exp \left(-c_{k, \alpha} \frac{|x - \varphi^{-t}(y)|^2}{\epsilon^2 t} \right),
    \end{equation}
    and $K_j$ satisfy the pointwise estimate
    \begin{equation}\label{eq:Kj_est}
        |\partial_x^\alpha K_j(x, y, t)| \le C_{j, \alpha} (\epsilon^2 t)^{-\frac{n + |\alpha|}{2}} \exp \left(-c_{j, \alpha} \frac{|x - \varphi^{-t}(y)|^2}{\epsilon^2 t} \right)
    \end{equation}
    where all constants are strictly positive and independent of $\epsilon$ and $t$. 
\end{lemma}
\begin{proof}
1. We proceed by induction. The base case estimates for $R_1(x, y, t)$ follows from Lemma~\ref{lem:R1_est}. Now assume the lemma holds for all $R_{j}$ with $j \le k - 1$, $k \ge 2$. By definition, 
\begin{equation}
\begin{aligned}
    R_k &= \partial_x^\alpha \int_0^t \int_{\R^n} R_1(x, z, t - s) R_{k - 1}(z, y, s)\, ds \\
    &= \int_0^{t/2} \int_{\R^n} R_1(x, z, t - s) R_{k - 1}(z, y, s)\, ds \\
    &\qquad \qquad + \int_{t/2}^t \int_{\R^n} R_1(x, z, t - s) R_{k - 1}(z, y, s)\, ds \\
    &=: R_{k1}(x, y, t) + R_{k2}(x, y, t).
\end{aligned}
\end{equation}
First, we estimate $R_{k1}$. It follows from the induction hypothesis and \eqref{eq:dxR1} that 
\begin{multline}\label{eq:R_j1}
    |\partial_x^\alpha R_{k1}(x, y, t)| \le C_\alpha \int_0^{t/2} \int_{\R^n} (\epsilon^2 (t - s))^{-\frac{n + \alpha}{2}} (\epsilon^2 s)^{-\frac{n}{2}} (t - s)^{-\ha} s^{\frac{k - 3}{2}} \\
    \exp\left(-c \frac{|x - \varphi^{-(t - s)}(z)|^2}{\epsilon^2(t - s)} \right) \exp\left(-c \frac{|z - \varphi^{-s}(y)|^2}{\epsilon^2s} \right)\, dz ds
\end{multline}
We first estimate the exponential term in the product inside the integrand. Observe that one of the exponential terms is localized near $\varphi^{t - s}(x)$ and the other is localized near $\varphi^{-s}(y)$. More precisely, for $0 \le s \le t \le 1$, we have
\begin{align}
    &\int_{\R^n} \exp \left( -c\frac{|x - \varphi^{-(t - s)}(z)|^2}{\epsilon^2(t - s)} \right) \exp \left(-c \frac{|z - \varphi^{-s}(y)|^2}{\epsilon^2 s} \right) dz \nonumber \\
    \le& \int_{\R^n} \exp \left( -c\frac{|\varphi^{t - s}(x) - z|^2}{\epsilon^2 (t - s)} \right) \exp \left(-c\frac{|z - \varphi^{-s}(y)|^2}{\epsilon^2 s} \right) dz \nonumber\\
    \le& \int_{\R^n} \exp \left( -c\frac{|(z - \varphi^{-s}(y)) - (\varphi^{t - s}(x) -\varphi^{s}(y))
    |^2}{\epsilon^2 (t - s)} \right) \exp \left(-c \frac{|z - \varphi^{-s}(y)|^2}{\epsilon^2 s} \right) dz \nonumber \\
    =& \exp \left(-c \frac{|x - \varphi^{-t}(y)|^2}{\epsilon^2 t} \right) \int_{\R^n} \exp\left(-c \left| \sqrt{\frac{t}{\epsilon^2 (t - s) s}} \tilde z - \sqrt{\frac{s}{\epsilon^2 (t - s) t}} w \right|^2 \right) d\tilde z \nonumber \\
    \le& C \left(\frac{t}{\epsilon^2 (t - s)s} \right)^{-\frac{n}{2}} \exp \left(-c \frac{|x - \varphi^{-t}(y)|^2}{\epsilon^2 t} \right) \label{eq:exp_integral}
\end{align}
where the constant $c > 0$ may change from line to line but remains independent of $\epsilon$ and $t$, and we made the substitutions $w := \varphi^{t - s}(x) - \varphi^s(y)$ and $\tilde z := z - \varphi^{-s}(y)$
Combining \eqref{eq:exp_integral} with \eqref{eq:R_j1}, we find that 
\begin{equation}\label{eq:Rk1_est}
\begin{aligned}
    |\partial_x^\alpha R_{k1}(x, y, t)| &\le C \exp \left(-c \frac{|x - \varphi^{-t}(y)|^2}{\epsilon^2 t} \right)\\
    &\qquad \int_0^{t/2} (\epsilon^2 t)^{-\frac{n}{2}} (\epsilon^2 (t - s))^{-\frac{|\alpha|}{2}} (t - s)^{-\ha} s^{\frac{k - 3}{2}} \, ds\\
    &\le C(\epsilon^2 t)^{-\frac{n + |\alpha|}{2}} t^{\frac{k}{2} - 1} \exp \left(-c \frac{|x - \varphi^{-t}(y)|^2}{\epsilon^2 t} \right).
\end{aligned}
\end{equation}

\noindent
2. Now we need the same estimate for $R_{k2}$. We need to make use of the oscillations in higher derivatives of $R_1$. Observe that
\begin{align}
    &\int_{\R^n} \partial_x^\alpha R_1(x, z, t - s) R_{k - 1}(z, y, s)\, dz \nonumber \\
    =& \int_{\R^n} (\partial_x^\alpha R)(x, \varphi^{t - s}(x - z), t - s) R_{k - 1}(\varphi^{t - s}(x - z), y, s) \det(d\varphi^{t - s}(x - z))\, dz \nonumber\\
    =& \int_{\R^n} \big[\partial_x^\alpha R_1(x, \varphi^{t - s}(x - z), t - s) + \partial_z^\alpha R_1(x, \varphi^{-t - s}(x - z), t - s\big] \nonumber\\
    &\qquad \qquad R_{k - 1}(\varphi^{t - s}(x - z), y, s) \det(d\varphi^{t - s}(x - z))\, dz\nonumber\\
    =& \int_{\R^n}\partial_x^\alpha R_1(x, \varphi^{t - s}(x - z), t - s) R_{k + 1}(\varphi^{t - s}(x - z), y, s) \det(d\varphi^{t - s}(x - z))\, dz \nonumber\\
    &\quad + (-1)^{|\alpha|}\int_{\R^n} R_1(x, \varphi^{t - s}(x - z), t - s) \nonumber\\
    &\qquad \qquad \qquad \qquad \partial_z^\alpha(R_{k + 1}(\varphi^{t - s}(x - z), y, s) \det(d\varphi^{t - s}(x - z)))\, dz \label{eq:gaussian_IBP}
\end{align}
Using \eqref{eq:dxR1'}, the first term on the right-hand-side of \eqref{eq:gaussian_IBP} can be estimated by 
\begin{align}
    &\left|\int_{\R^n}\partial_x^\alpha R_1(x, \varphi^{t - s}(x - z), t - s) R_{k + 1}(\varphi^{t - s}(x - z), y, s) \det(d\varphi^{t - s}(x - z))\, dz\right| \nonumber \\
    \le& C \int_{t/2}^t \int_{\R^n} (\epsilon^2(t - s))^{-\frac{n}{2}} (\epsilon^2 s)^{-\frac{n}{2}} (t - s)^{-\ha} s^{\frac{k - 3}{2}} \nonumber \\
    &\qquad \qquad \qquad  \exp \left( - c \frac{|z|^2}{\epsilon^2(t - s)} \right) \exp \left( -c\frac{|\varphi^{t - s}(x - z) - \varphi^{-s}(y)|^2}{\epsilon^2 s} \right)\, dz ds \nonumber \\
    \le& C \exp\left(-c \frac{|x - \varphi^{-t}(y)|^2}{\epsilon^2 t} \right) \int_{t/2}^t (\epsilon^2 t)^{-\frac{n}{2}} (t - s)^{-\ha} s^{\frac{k - 3}{2}}\, ds \nonumber \\
    \le& C (\epsilon^2 t)^{-\frac{n}{2}} t^{\frac{k}{2} - 1} \exp\left(-c \frac{|x - \varphi^{-t}(y)|^2}{\epsilon^2 t} \right) \label{eq:Rk2_good}
\end{align}
Here, the integral of the product of two Gaussians is computed similarly to \eqref{eq:exp_integral}, and the constants $C$ and $c > 0$ may change from line to line, but remains independent of $\epsilon$ and $t$. 

On the other hand, the second term on the right-hand-side of \eqref{eq:gaussian_IBP} is bounded by 
\begin{align}
    &\left|\int_{\R^n} R_1(x, \varphi^{t - s}(x - z), t - s) \partial_z^\alpha(R_{k + 1}(\varphi^{t - s}(x - z), y, s) \det(d\varphi^{t - s}(x - z)))\, dz\right| \nonumber \\
    \le& C \int_{t/2}^t \int_{\R^n} (\epsilon^2 (t - s))^{-\frac{n}{2}} (\epsilon^2 s)^{-\frac{n + |\alpha|}{2}} (t - s)^{-\ha} s^{\frac{k - 3}{2}} \nonumber \\
    &\qquad \qquad \qquad  \exp \left( - c \frac{|z|^2}{\epsilon^2 (t - s)} \right) \exp \left( -c\frac{|\varphi^{t - s}(x - z) - \varphi^{-s}(y)|^2}{\epsilon^2 s} \right)\, dz ds \nonumber \\
    \le& C \exp\left(-c \frac{|x - \varphi^{-t}(y)|^2}{\epsilon^2 t} \right) \int_{t/2}^t (\epsilon^2 t)^{-\frac{n}{2}} (\epsilon^2 s)^{-\frac{|\alpha|}{2}} (t - s)^{-\ha} s^{\frac{k - 3}{2}}\, ds \nonumber \\
    \le& C (\epsilon^2 t)^{-\frac{n + |\alpha|}{2}} t^{\frac{k}{2} - 1} \exp\left(-c \frac{|x - \varphi^{-t}(y)|^2}{\epsilon^2 t} \right). \label{eq:Rk2_bad}
\end{align}
Combining~\eqref{eq:Rk2_good} and~\eqref{eq:Rk2_bad}, we obtain the desired estimate on $R_{k2}$:
\begin{equation*}
    |\partial_x^\alpha R_{k 2}(x, y, t)| \le C(\epsilon^2 t)^{-\frac{n + |\alpha|}{2}} t^{\frac{k}{2} - 1} \exp\left(-c \frac{|x - \varphi^{-t}(y)|^2}{\epsilon^2 t} \right).
\end{equation*}
Combining with~\eqref{eq:Rk1_est} yields~\eqref{eq:Rk_est}.

\noindent
3. It remains to estimate $K_j$. Note that 
\begin{equation}
    \begin{gathered}
        |\partial_x^\alpha K_0(x, y, t)| \le C (\epsilon^2 t)^{-\frac{n + |\alpha|}{2}} \exp \left(-c \frac{|x - \varphi^{-t}(y)|^2}{\epsilon^2 t} \right) \\
        |\partial_x^\alpha K_0(x, \varphi^t(x - y), t)| \le C (\epsilon^2 t)^{-\frac{n}{2}} \exp \left(-c \frac{|x - \varphi^{-t}(y)|^2}{\epsilon^2 t} \right) 
    \end{gathered}
\end{equation}
Repeating Step 2 replacing $R_1(x, z, t - s)$ with $K_0(x, z, t - s)$ yields the estimate
\begin{equation*}
    \left|\int_0^t \int_{\R^n} \partial_x^\alpha K_0(x, z, t - s) R_{k}(z, y, s)\, dz dt \right| \le C (\epsilon^2 t)^{-\frac{n + |\alpha|}{2}} t^{\frac{k}{2}} \exp \left(-c \frac{|x - \varphi^{-t}(y)|^2}{\epsilon^2 t} \right)
\end{equation*}
Summing over $k$ yields
\begin{equation}
    |\partial_x^\alpha K_j(x, y, t)| \le (\epsilon^2 t)^{-\frac{n + |\alpha|}{2}} \exp \left(-c \frac{|x - \varphi^{-t}(y)|^2}{\epsilon^2 t} \right)
\end{equation}
for $\epsilon, t \in (0, 1]$. In other words, the effects of $K_0$ dominates as expected. 
\end{proof}

Estimates for derivatives in $y$ for $R_k(x, y, t)$ can also be established, but we really only need one derivative in $y$ for $R_1$ and $R_2$, so we focus on these cases for simplicity. 
\begin{lemma}\label{lem:R1_est_y}
    $R_1(x, y, t)$ and $R_2(x, y, t)$ satisfy the estimates
    \begin{equation}
        |\nabla_y R_1(x, y, t)| \le C (\epsilon^2 t)^{-\frac{n + 1}{2}} t^{-\ha} \exp\left(-c \frac{|x - \varphi^{-t}(y)|^2}{\epsilon^2 t} \right)
    \end{equation}
    and 
    \begin{equation}\label{eq:dyR2}
        |\nabla_y R_2(x, y, t)| \le C (\epsilon^2 t)^{-\frac{n + 1}{2}} \exp\left(-c \frac{|x - \varphi^{-t}(y)|^2}{\epsilon^2 t} \right)
    \end{equation}
    for $\epsilon, t \in (0, 1]$ with $c, C > 0$ independent of $\epsilon$ and $t$. 
\end{lemma}
\begin{proof}
    1. The estimate for $\nabla_yR_1(x, y, t)$ follows from Lemma \ref{lem:R1_est} and the identity
    \begin{multline}
        \nabla_x R_1(x, \varphi^t(x - y), t) = (\nabla_x R_1)(x, \varphi^t(x - y), t) \\
        - (d\varphi^t(x - y))^\intercal (\nabla_yR_1)(x, \varphi^t(x - y), t).
    \end{multline}
    Indeed, it follows that 
    \begin{equation}
        |(\nabla_y R_1)(x, \varphi^t(x - y), t)| \le C((\epsilon^2 t)^{-\frac{n + 1}{2}} t^{-\ha} + (\epsilon^2 t)^{-\frac{n}{2}} t^{-\ha})\exp\left(-c \frac{|y|^2}{\epsilon^2 t} \right),
    \end{equation}
    so the desired estimate follows upon putting $y \mapsto x - \varphi^{-t}(y)$.
    
    \noindent
    2. Estimating $\nabla_y R_2(x, y, t)$ is a similar argument to the proof of Lemma \ref{lem:Rk_est}. First, observe that 
    \begin{align}
        &\left|\int_{t/2}^t \int_{\R^n} R_1(x, z, t - s) \nabla_y R_1 (z, y, s) \, dz ds \right| \nonumber \\
        \le& C\int_{t/2}^t \int_{\R^n} (h(t - s))^{-\frac{n}{2}} (\epsilon^2 s)^{-\frac{n + 1}{2}} (t - s)^{-\ha} s^{-\ha} \nonumber \\
        &\hspace{2cm} \exp \left(-c \frac{|x - \varphi^{-(t - s)}(y)|^2}{\epsilon^2 (t - s)} \right) \exp\left( -c \frac{|x - \varphi^{-s}(y)|^2}{\epsilon^2 s} \right) \, dz ds \nonumber \\
        \le& C (\epsilon^2 t)^{-\frac{n}{2}} \exp \left(-c\frac{|x - \varphi^{-t}(y)|^2}{\epsilon^2 t} \right) \int_{t/2}^t (t - s)^{-\ha} s^{-\ha} (\epsilon^2 s)^{-\ha}\, ds \nonumber \\
        \le& C (\epsilon^2 t)^{-\frac{n + 1}{2}} \exp \left(-c\frac{|x - \varphi^{-t}(y)|^2}{\epsilon^2 t} \right)\label{eq:dyR2_1}
    \end{align}
    Next, we need the integral from $0$ to $t/2$, which requires requires the following integration by parts:
    \begin{align}
        &\int_{\R^n} R_1(x, z, t - s) \nabla_y R_1(z, y, s)\, dz \nonumber \\
        =& \int_{\R^n} R_1(x, z + \varphi^{-s}(y), t - s) (\nabla_y R_1)(z + \varphi^{-s}(y), y, s)\, dz \nonumber \\
        =& \int_{\R^n} R_1(x, z + \varphi^{-s}(y), t - s)[\nabla_y R_1(z + \varphi^{-s}(y), y, s) \nonumber \\
        &\hspace{6cm} - (d\varphi^{-s}(y))^\intercal \nabla_z R_1(z + \varphi^{-s}(y), y, s)]\, dz \nonumber \\
        =& \int_{\R^n} R_1(x, z + \varphi^{-s}(y), t - s) \nabla_y R_1(z + \varphi^{-s}(y), y, s)\, dz \nonumber \\
        &\qquad + \int_{\R^n} d\varphi^{-s}(y) \nabla_z R_1(x, z + \varphi^{-s}(y), t - s) R_1(z + \varphi^{-s}(y), y, s)\, dz. \label{eq:R1dyR1}
    \end{align}
    Integrating the first term on the right-hand-side of \eqref{eq:R1dyR1} from $0$ to $t/2$ and estimating the derivatives using Lemma \ref{lem:R1_est}, we see that 
    \begin{align}
        &\left|\int_0^{t/2}\int_{\R^n} R_1(x, z + \varphi^{-s}(y), t - s) \nabla_y R_1(z + \varphi^{-s}(y), y, s)\, dz ds \right| \nonumber \\
        \le& C \int_0^{t/2}\int_{\R^n}(\epsilon^2 (t - s))^{-\frac{n}{2}} (t -s )^{-\ha} (\epsilon^2 s)^{\frac{n}{2}} s^{-\ha} \nonumber \\
        &\hspace{2cm} \exp \left(-c \frac{|x - \varphi^{-(t - s)}(z - \varphi^{-s}(y))|^2}{\epsilon^2 (t - s)} \right) \exp \left(-c \frac{|z|^2}{\epsilon^2 s} \right)\, dz ds \nonumber \\
        \le& C (\epsilon^2 t)^{-\frac{n}{2}} \exp \left(-c \frac{|x - \varphi^{-t}(y)|^2}{\epsilon^2 t} \right) \int_0^{t/2} (t - s)^{-\ha} s^{-\ha}\, ds \nonumber \\
        \le& C (\epsilon^2 t)^{-\frac{n}{2}}  \exp \left(-c \frac{|x - \varphi^{-t}(y)|^2}{\epsilon^2 t} \right). \label{eq:dyR2_2}
    \end{align}
    Now integrating the second term on the right-hand-side of \eqref{eq:R1dyR1}, we see that
    \begin{align}
        & \left|\int_0^{t/2} \int_{\R^n} d\varphi^{-s}(y) \nabla_z R_1(x, z + \varphi^{-s}(y), t - s) R_1(z + \varphi^{-s}(y), y, s)\, dz \right|\nonumber \\
        \le& C \int_0^{t/2}\int_{\R^n}(\epsilon^2 (t - s))^{-\frac{n + 1}{2}} (t -s )^{-\ha} (\epsilon^2 s)^{\frac{n}{2}} s^{-\ha} \nonumber \\
        &\hspace{2cm} \exp \left(-c \frac{|x - \varphi^{-(t - s)}(z - \varphi^{-s}(y))|^2}{\epsilon^2 (t - s)} \right) \exp \left(-c \frac{|z|^2}{\epsilon^2 s} \right)\, dz ds \nonumber \\
        \le& C (\epsilon^2 t)^{-\frac{n}{2}} \int_0^{t/2} (t - s)^{-\ha} (\epsilon^2 (t -s))^{-\ha} s^{-\ha}\, ds  \exp \left(-c \frac{|x - \varphi^{-t}(y)|^2}{\epsilon^2 t} \right) \nonumber \\
        \le& C (\epsilon^2 t)^{-\frac{n + 1}{2}}  \exp \left(-c \frac{|x - \varphi^{-t}(y)|^2}{\epsilon^2 t} \right) \label{eq:dyR2_3}
    \end{align}
    Combining \eqref{eq:R1dyR1}-\eqref{eq:dyR2_3} gives
    \begin{equation*}
        \left| \int_0^{t/2} \int_{\R^n} R_1(x, z, t - s) \nabla_y R_1(z, y, s)\, ds\right| \le C (\epsilon^2 t)^{-\frac{n + 1}{2}}  \exp \left(-c \frac{|x - \varphi^{-t}(y)|^2}{\epsilon^2 t} \right),
    \end{equation*}
    which yields \eqref{eq:dyR2} upon combining with \eqref{eq:dyR2_1}. 
\end{proof}

The upshot of the remainder estimates of Lemmas~\ref{lem:Rk_est} and~\ref{lem:R1_est_y} is that the $L^1_x$-norm of $\epsilon$-semiclassical derivatives of $R_k(x, y, t)$ is controlled uniformly in $\epsilon$ and $y$, and improves in powers of $t$ for large $k$. This uniformity in $\epsilon$ is crucial for establishing an $L^1$-based smoothing estimate on semiclassical derivatives of solutions to the evolution equation~\eqref{eq:cauchy_problem}.

\section{\texorpdfstring{$L^1$}{L1} classical estimates}\label{sec:L1}
Now we use the parametrix constructed in the previous section to establish the necessary $L^1$ estimates. We first collect some classical results on the well-posedness of second order parabolic equations with unbounded coefficients. 
\begin{proposition}\label{prop:WP}
    Consider the Cauchy problem~\eqref{eq:cauchy_problem} satifying assumptions~\eqref{eq:uniform_ell_est}-\eqref{eq:v_est}. If $u_0 \in C^\infty(\R^n)$ satisfies 
    \begin{equation}
        |u_0(x)| \le B e^{\beta |x|^2}
    \end{equation}
    for some $B, \beta > 0$. Then for each $\epsilon > 0$, there exists $T_\epsilon > 0$ and a unique classical solution to \eqref{eq:cauchy_problem} for $0 \le t \le T_\epsilon$ that satisfies
    \begin{equation}
        |u(x, t)| \le B_\epsilon e^{\beta_\epsilon |x|^2}
    \end{equation}
    for some $B_\epsilon, \beta_\epsilon > 0$. Furthermore, $u(x, t)$ is given by 
    \begin{equation}
        e^{tQ}u_0(x, t) := \int_{\R^n} K(x, y, t; \epsilon) u_0(y)\, dy
    \end{equation}
    for some $K(\bullet, \bullet, \bullet; \epsilon) \in \mathcal C^\infty(\R^n_x \times \R^n_y \times \R_t)$ such that
    \begin{equation}
        K(x, y, t; \epsilon) \le C_\epsilon \exp \left(-c_\epsilon \frac{|x - y|^2}{t} \right), \quad 0 < t \le T_\epsilon
    \end{equation}
    for some $C_\epsilon, c_\epsilon > 0$ and $e^{tQ}$ satisfies the semigroup property.
\end{proposition}
The fundamental solution in the case of unbounded coefficients was first constructed in \cite{Zhito}. We refer the reader to \cite[\S2.2]{Ejdelman94} and \cite[\S2.4]{Friedman_parabolic} for proofs of the above proposition and a complete introduction to well-posedness theory for the parabolic Cauchy problem with unbounded coefficients. 

We also remark that the dependence on $\epsilon$ for the time of existence does not matter for our setting due to the semigroup property; having sufficiently good estimates on the solution at later times will allow us to extend the time of existence.

\subsection{Short-time estimate}
Now we make use of the crucial assumption \eqref{eq:divfree} that $v(x)$ is divergence-free. This gives us the conservation of mass for the evolution equation~\eqref{eq:cauchy_problem}.
\begin{lemma}\label{lem:L1_bound}
    Suppose $u(x, t) \in C^\infty(\R^n_x \times \R_t)$ satisfies \eqref{eq:cauchy_problem} and 
    \[|u(x, t)| \le B e^{\beta|x|^2}\]
    for some $B, \beta > 0$. Then 
    \begin{equation}
        \|u(t)\|_{L^1} \le \|u_0\|_{L^1}. 
    \end{equation}
\end{lemma}
\begin{proof}
By Proposition \eqref{prop:WP}, we know that $u(\bullet, t) \in L^1$. The Fokker--Planck equation preserves positivity \cite[\S4.1 Theorem 9]{Friedman_parabolic}. Since $v(x)$ is divergence free, the $L^1$-boundedness of the Fokker-Planck evolution follows by decomposing $u$ into positive and negative parts by $u_0 = u_+ - u_-$, and seeing that
\[\|e^{tQ}u_0\|_{L^1} \le \|e^{Qt} u_+\|_{L^1} + \|e^{Qt} u_-\|_{L^1} = \|u_+ + u_-\|_{L^1} = \|u_0\|_{L^1}\]
as desired. 
\end{proof}

\begin{lemma}\label{lem:dyK_est}
    Let $K(x, y, t)$ be the fundamental solution from Proposition~\ref{prop:WP}. Then 
    \begin{equation}
        \|\partial_y K(\bullet, y, t)\|_{L^1} \le C(\epsilon^2t)^{-\ha},
    \end{equation}
    where $C$ is independent of $0 < \epsilon, t \le 1$.
\end{lemma}
\begin{proof}
    Using Duhamel's formula, we see that 
    \begin{equation}
        \partial_y K(x, y, t) = \partial_y K_1(x, y, t) + \int_0^{t} e^{(t - s)Q} \partial_y R_2(x, y, s)\, ds
    \end{equation}
    From Lemma \ref{lem:R1_est_y}, we see that 
    \begin{equation}
        \|\partial_y K_1(\bullet, y, t)\|_{L^1} \le C(\epsilon^2 t)^{-\ha}
    \end{equation}
    and 
    \begin{equation}
        \|\partial_y R_2(\bullet, y, t)\|_{L^1} \le C(\epsilon^2 t)^{-\ha}.
    \end{equation}
    Therefore
    \begin{equation}
        \|\partial_y K(\bullet, y, t)\|_{L^1} \le C(\epsilon^2 t)^{-\ha} + \int_0^t (\epsilon^2 s)^{-\ha}\, ds \le C(\epsilon^2 t)^{-\ha} 
    \end{equation}
    as desired. 
\end{proof}
Let $W^{r, 1}(\R^n)$ denote the $L^1$-based order $r$ Sobolev space. We now show that derivatives of a solution to \eqref{eq:cauchy_problem} up to a $\mathcal O(1)$ amount of time independent of $\epsilon$ is controlled by derivatives of the initial data. 
\begin{proposition}\label{prop:short_time}
    Suppose $u(x, t) \in C^\infty(\R^n_x \times \R_t)$ satisfies \eqref{eq:cauchy_problem} and 
    \[|u(x, t)| \le B e^{\beta|x|^2}\]
    for some $B, \beta > 0$. Then for every $r \in \N$, there exists $0 < \tau_r\le 1$ independent of $0 < \epsilon \le 1$ such that for all $0 \le t \le \tau_r$, $u$ satisfies the estimate
    \begin{equation}
        \|u(t)\|_{W^{r, 1}} \le C_\alpha \|u_0\|_{W^{r, 1}}.
    \end{equation}
\end{proposition}
\begin{proof}
For $r = 0$, Lemma \ref{lem:L1_bound} gives the desired $L^1$ estimate. We handle higher derivatives inductively. Assume that the lemma holds for all $\tilde r < r$. Differentiating \eqref{eq:cauchy_problem}, we see that 
\begin{equation}\label{eq:Ws_duhamel}
    (\partial_t - Q)(\partial_x^\alpha u) = h \nabla \cdot [\partial_x^\alpha, A(x)] \nabla u + [\partial_x^\alpha, v(x) \cdot \nabla]u, \qquad |\alpha| = r.
\end{equation}
Using Duhamel's formula, we can write
\begin{equation}
    \partial_x^\alpha u = e^{tQ}(\partial_x^\alpha u_0) + \int_0^t e^{(t - s)Q} (h \nabla \cdot [\partial_x^\alpha, A(x)] \nabla u(s) + [\partial_x^\alpha, v(x) \cdot \nabla] u(s))\, ds
\end{equation}
By assumption $v(x)$, satisfies~\ref{eq:v_est}, which means that $[\partial_x^\alpha, v(x) \cdot \nabla]$ is an order $r$ differential operator with uniformly bounded coefficients. Therefore, by Lemma~\ref{lem:L1_bound}, we have
\begin{equation}\label{eq:Ws_duhamel1}
    \left\|\int_0^t e^{(t - s) Q} [\partial_x^\alpha, v(x) \cdot \nabla] u(s)\, ds \right\|_{L^1} \le C t \max_{0 \le s \le t} \|u(s)\|_{W^{r, 1}}.
\end{equation}
The remaining term in the integrand of~\eqref{eq:Ws_duhamel} is handled using Lemma~\ref{lem:dyK_est}. Integrating by parts, we wee that 
\begin{align*}
    &\int_0^t e^{(t - s)Q} (\nabla \cdot [\partial_x^\alpha, A(x)] \nabla u(s))\, ds \\
    &= \int_0^t \int_{\R^n} K(x, y, t - s) (\nabla_y \cdot [\partial_y^\alpha, A(y)] \nabla_y u(y, s))\, dy ds \\
    &\le \int_0^t \int_{\R^n} \nabla_y K(x, y, t - s) \cdot ([\partial_y^\alpha, A(y)] \nabla_y u(y, s))\, dy ds
\end{align*}
Then it follows from Lemma~\ref{lem:dyK_est} and assumption~\eqref{eq:A_est} on $A$ that
\begin{equation}\label{eq:Ws_duhamel2}
    \left\|\int_0^t e^{(t - s)Q} (\epsilon^2 \nabla \cdot [\partial_x^\alpha, A(x)] \nabla u(s))\, ds \right\|_{L^1} \le C (\epsilon^2 t)^\ha \max_{0 \le s \le t} \|u\|_{W^{r, 1}}. 
\end{equation}
Combining~\eqref{eq:Ws_duhamel1} and~\eqref{eq:Ws_duhamel2} with~\eqref{eq:Ws_duhamel}, we see that 
\begin{equation*}
    \|\partial_x^\alpha u(t)\|_{L^1} \le \|\partial_x^\alpha u_0\|_{L^1} + C(t + (\epsilon^2 t)^\ha) \max_{0 \le s \le t} \|u(s)\|_{W^{r, 1}}.
\end{equation*}
By the induction hypothesis, we then have
\begin{equation}\label{eq:max}
    \|u(t)\|_{W^{r, 1}} \le \|u_0\|_{W^{r, 1}} + C(t + (\epsilon^2 t)^\ha) \max_{0 \le s \le t} \|u(s)\|_{W^{r, 1}}
\end{equation}
Taking $t$ sufficiently small independently of $\epsilon \in (0, 1]$, we obtain the desired estimate from~\eqref{eq:max}. 
\end{proof}

\subsubsection{Smoothing estimate}
In the Duhamel argument of the previous section, we are not able to push past time $t \simeq 1$. To obtain long time estimates, the other ingredient we need is that the $L^1$ norm of $\epsilon$-semiclassical derivatives of the solution at time $t \simeq 1$ is controlled by the $L^1$ norm at $t = 0$, $h = \epsilon^2$. 
\begin{proposition}\label{prop:smoothing}
    Suppose $u(x, t) \in C^\infty(\R^n_x \times \R_t)$ satisfies \eqref{eq:cauchy_problem} and 
    \[|u(x, t)| \le B e^{\beta|x|^2}\]
    for some $B, \beta > 0$. Then for each multiindex $\alpha \in \N_0^n$, there exists $0 < \tau_{\alpha} \le 1$ independent of $\epsilon$ such that 
    \begin{equation}
        \|(\epsilon \partial_x)^\alpha u(t)\|_{L^1} \le C \|u_0\|_{L^1}. 
    \end{equation}
    for all $t \ge \tau_\alpha$ and $0 < \epsilon \le 1$. 
\end{proposition}
Note that this captures a smoothing effect since higher regularity of the solution at time $t \simeq 1$ is controlled by just the mass of the initial data after $\mathcal O(1)$ amount of time. 
\begin{proof}
Let $K(x, y, t)$ denote the Schwartz kernel of $e^{tQ}$. By Duhamel's formula, we have
\begin{equation}\label{eq:differentiated_duhamel}
    (\epsilon \partial_x)^\alpha K(x, y, t) = (\epsilon \partial_x)^\alpha K_j(x, y, t) - \int_0^t (\epsilon \partial_x)^\alpha e^{(t - s)Q} R_{j + 1}(x, y, s)\, ds.
\end{equation}
By Lemma \ref{lem:Rk_est}, 
\begin{equation}
    \sup_{y \in \R^n} \|(\epsilon \partial_x)^\alpha K_j(\bullet, y, t)\| \le C t^{-\frac{|\alpha|}{2}}
\end{equation}
and 
\begin{equation}
    \sup_{y \in \R^n} \|(\epsilon \partial_x)^\alpha R_{j + 1}(\bullet, y, s)\| \le C s^{-\frac{|\alpha|}{2}} s^{\frac{j - 1}{2}}.
\end{equation}
By Lemma \ref{prop:short_time}, there exists $0 < \tau_\alpha < 1$ indepenent of $h$ such that for all $0 \le s \le t \le \tau_\alpha$, 
\begin{equation}
    \|(\epsilon \partial_x)^\alpha e^{(t - s)Q} R_{j + 1}(\bullet, y, s)\|_{L^1} \le C\sum_{\beta \le |\alpha|} \|(\epsilon \partial_x)^\beta R_{j + 1}(\bullet, y, s)\|_{L^1} \le C s^{\frac{ j - |\alpha| - 1}{2}}.
\end{equation}
Therefore, taking $j \ge |\alpha|$ so that the integral in \eqref{eq:differentiated_duhamel} converges, we find that 
\begin{equation}
    \sup_{y \in \R^n} \|(\epsilon \partial_x)^\alpha K(\bullet, y, t)\| \le C t^{-\frac{|\alpha|}{2}},
\end{equation}
which implies the lemma. 
\end{proof}

Combining Proposition \ref{prop:short_time} and \ref{prop:smoothing}, we have the following corollary stated in terms of $L^1$-based semiclassical Sobolev space $W^{r, 1}_\epsilon(\R^n)$, which is equivalent to $W^{r, 1}(\R^n)$ as a set, but is equipped with norm 
\begin{equation}
    \|u\|_{W^{r, 1}_\epsilon} := \sum_{|\alpha| \le r} \|(\epsilon \partial_x)^\alpha u\|_{L^1},
\end{equation}
where the derivatives are understood in the distributional sense. 
\begin{corollary}\label{cor:semiclassical}
    Suppose $u(x, t) \in C^\infty(\R^n_x \times \R_t)$ satisfies \eqref{eq:cauchy_problem} and 
    \[|u(x, t)| \le B e^{\beta|x|^2}\]
    for some $B, \beta > 0$. Then for all $r \ge 0$ and $0 < \epsilon \le 1$,
    \begin{equation}
        \|u(t)\|_{W^{r,1}_\epsilon} \le C_r \|u_0\|_{W^{r,1}_\epsilon}. 
    \end{equation}
    for all $t \ge 0$.
\end{corollary}

\subsection{An example} The Fokker--Planck equation can be solved explicitly when the Hamiltonian is quadratic and the Lindbladians are linear for Gaussian initial data-- see for instance \cite{HRRb}. This gives simple examples that we can compute, and we see that in general, we cannot do better than Lemma~\ref{prop:smoothing}. We consider an example in $\R^2_{x, \xi}$ given by 
\begin{equation}\label{eq:example_Q}
    Q = \epsilon^2 \Delta + x \partial_x - \xi \partial_\xi.
\end{equation}
This corresponds to the jump functions $\ell_1(x, \xi) = x$ and $\ell_2(x, \xi)= \xi$, and the Hamiltonian $p(x, \xi) = x \xi$. Consider the initial data
\begin{equation*}
    u_0(x) = \frac{1}{h} \exp \left(-\frac{x^2 + \xi^2}{h} \right).
\end{equation*}
The equation preserves Gaussians as well as the symmetry about the $x$ and $\xi$ axis. Therefore, the solution to \eqref{eq:cauchy_problem} must be of the form
\begin{equation*}
    u(x, \xi, t) = (a(t) b(t))^{-\ha}\exp\left(- \frac{x^2}{a(t)} - \frac{\xi^2}{b(t)} \right).
\end{equation*}
Then solving for $a(t)$ and $b(t)$, we find
\begin{equation*}
a(t) = (h - 2 \epsilon^2) e^{-2t} + 2 \epsilon^2, \qquad b(t) = (h + 2 \epsilon^2) e^{2t} - 2 \epsilon^2.
\end{equation*}
See Figure~\ref{fig:hyperbolic}.
\begin{figure}
    \centering
    \includegraphics{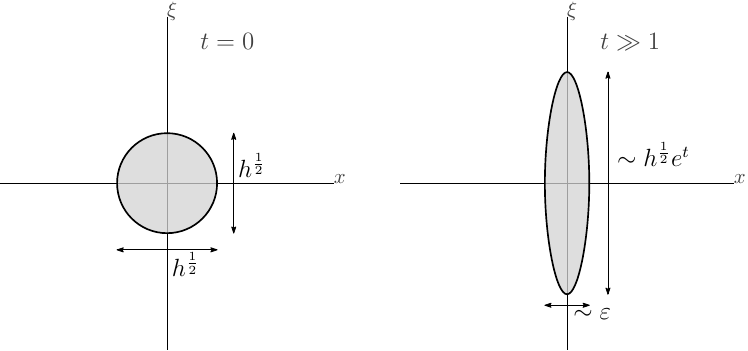}
    \caption{Fokker--Planck evolution of the standard coherent state centered at $0$ by \eqref{eq:example_Q} in the regime $h \le \epsilon^2 \le 1$.}
    \label{fig:hyperbolic}
\end{figure} 
In particular, observe that for $t \ge 1$, $u$ oscillates on scale $\epsilon$ in the $x$ direction if $h \le \epsilon^2 \le 1$ (which corresponds to the regime of Theorem \ref{thm:gaussian_case}). In terms of $L^1$ estimates, we see that
\begin{equation*}
    \|(\epsilon \partial_x)^k u(x, \xi, t)\|_{L^1_{x, \xi}(\R^2)} = C_k \epsilon^k \alpha(t)^{-\frac{k}{2}} \simeq C_k
\end{equation*}
From this example, we see that the smoothing estimate~\eqref{prop:smoothing} is optimal in $\epsilon$ and $t$.

\section{Classical--Quantum correspondence}
We return to the Lindblad equation. Recall that under our assumptions~\eqref{eq:FP_def} for $P$ and $L_j$, we can rewrite the Lindbladian as 
\begin{equation}
    \mathcal L A = \frac{i}{2h}[P, A] - \frac{\gamma}{h} \sum_{j = 1}^J [L_j, [L_j, A]]. 
\end{equation}
We now compute the asymptotic expansion of the involved operator compositions, from which the Fokker--Planck operator will appear. It follows from Lemma~\ref{lem:composition} that
\begin{equation*}
    [P,A] = \Op_h^\w(-ih H_p a + h^{3 - 3\rho}e) \quad \text{for some} \quad e \in S^{L^1}_\rho
\end{equation*}
We emphasize here that this commutator relation is specific to the Weyl quanitization; the order $h^{2 - 2\rho}$ terms cancel so that we end up with the better $h^{3 - 3 \rho}$ remainder. Applying the commutator identity twice, we also have
\begin{equation*}
    [L_j, [L_j, A]] = \Op^\w_h (-h^2 H_{\ell_j}^2 a) + h^{4 - 4\rho} \Op_h^\w(e) \quad \text{for some} \quad e \in S_\rho^{L^1}.
\end{equation*}
Therefore, we see that 
\begin{align}
    \mathcal L A &= \Op_h^\w \Big(H_p a + \frac{\gamma h}{2} \sum_{j = 1}^J H_{\ell_j}^2 a \Big) + h^{2 - 3 \rho} (1 + h^{1 - \rho} \gamma) \Op_h^\w(e) \nonumber \\
    &= \Op_h^\w(Qa) + h^{2 - 3 \rho} (1 + h^{1 - \rho} \gamma) \Op_h^\w(e) \label{eq:L_expansion}
\end{align}
for some $e \in S^{L^1}_\rho$, where $Q$ is the Fokker--Planck operator defined in~\eqref{eq:FP_def}. To ensure that the error term in~\eqref{eq:L_expansion} blow up with $h$ up to the strong coupling regime $\gamma = h^{-1}$, we see that it is reasonable to assume $0 \le \rho \le \ha$. See Remark after proof of Theorem~\ref{thm:general} for more details on the range of $\rho$.

\subsection{Trace norm estimate}
We first show that the classical evolution given by the Fokker--Planck equation agrees with the quantum evolution given by the Lindblad equation for a long time. 
\begin{theorem}\label{thm:general}
    Suppose $\mathcal L$ given in~\eqref{eq:full_lindbladian} satisfies~\eqref{eq:PL_assumptions} and~\eqref{eq:jump_elliptic}. Further assume that $h^{2 \rho - 1} \le \gamma \le h^{-1}$ and $0 \le \rho \le \ha$. If $A(t)$ satisfies 
    \begin{equation*}
        \partial_t A(t) = \mathcal L A(t), \quad A(0) = \Op_h^\w(a_0), \quad a_0 \in S^{L^1}_\rho,
    \end{equation*}
    then there exists $a(t)$ that satisfies
    \begin{equation*}
        \partial_t a(t) = Qa(t), \qquad a(0) = a_0
    \end{equation*}
    such that 
    \begin{equation}\label{eq:tr_est}
        \|A(t) - \Op_h^\w(a(t))\|_{\tr} \le \begin{cases} Ct(h^{2 - 3 \rho} + h^{3 - 4 \rho} \gamma) & 0 \le t \le 1 \\ C\big[h^{2 - 3 \rho} + h^{3 - 4 \rho} \gamma + t (h^{\ha} \gamma^{-\frac{3}{2}} + h \gamma^{-1}) \big] & t \ge 1 \end{cases}
    \end{equation}
\end{theorem}
\begin{proof}
Consider initial data $a_0 \in S^{L^1}_\rho$, and let $a(t):= e^{tQ} a_0$. Let $\epsilon^2 := \frac{\gamma h}{2}$. It follows from Corollary~\ref{cor:semiclassical} that for $h^\rho \le \epsilon \le 1$, we have
\begin{align*}
    h^{-n} \|\partial^\alpha (e^{tQ} a_0)\|_{L^1} &\le C_\alpha h^{-n} \epsilon^{-|\alpha|} \sum_{|\beta| \le |\alpha|} \|(\epsilon \partial)^\alpha a_0\|_{L^1} \\
    &\le C_\alpha h^{-\rho |\alpha|} \sum_{k \le |\alpha|} \epsilon^{-(|\alpha| - k)} h^{\rho(|\alpha| - k)} \\
    &\le C_{\alpha} h^{-\rho |\alpha|}.
\end{align*}
Therefore, 
\begin{equation}
    a(t) \in S^{L^1}_\rho, \qquad t \ge 0,
\end{equation}
uniformly in $t$ and $h$. For $t \ge 1$, it follows from Proposition~\eqref{prop:smoothing} that
\begin{equation*}
    h^{-n} \|\partial^\alpha (e^{tQ} a_0)\|_{L^1} \le C_\alpha h^{-n} \epsilon^{-|\alpha|} \|a_0\|_{L^1} \le C_\alpha \epsilon^{-|\alpha|}, 
\end{equation*}
so 
\begin{equation}
    a(t) \in S^{L^1}_{\tilde \rho}  \quad \text{where} \quad \tilde \rho = \frac{\log \epsilon}{\log h}, \quad t \ge 1.
\end{equation}
Note that $\tilde \rho \le \rho$. By \eqref{eq:L_expansion}, we have 
\begin{equation}\label{eq:e1}
    \partial_t \Op_h^\w(a(t)) = \Op_h^\w(Qa(t)) = \mathcal L \Op_h^\w(a(t)) + \Op_h^\w(e_1(t)) 
\end{equation}
where the error $e_1$ satisfies
\begin{equation}\label{eq:e1_est}
e_1(t) \in 
    \begin{cases}
        h^{2 - 3 \rho} (1 + h^{1 - \rho} \gamma) S^{L_1}_\rho, & t \ge 0 \\
        h^{2 - 3 \tilde \rho} (1 + h^{1 - \tilde \rho} \gamma) S^{L_1}_{\tilde \rho}, & t \ge 1
    \end{cases}
\end{equation}
Using Duhamel's formula, we have
\begin{equation}
    A(t) = \Op_h^\w(a(t)) + \int_0^t e^{(t - s) \mathcal L} \Op_h^\w(e_1(t)),
\end{equation}
where $e^{(t - s) \mathcal L}$ is a well-defined completely positive contraction on the space of trace class operators -- see \cite{Davies77} and \cite{GZ24}. Therefore, since $0 \le \rho \le \ha$, we have 
\begin{equation}\label{eq:split}
\begin{aligned}
    \|A(t) - \Op_h^\w(a(t)) \|_{\tr} &\le \int_0^1 \|\Op_h^\w(e_1(s))\|_{\tr}\, ds + \int_1^t \|\Op_h^\w(e_1(s))\|_{\tr} \, ds \\
    &\le C h^{2 - 3 \rho}(1 + h^{1 - \rho}\gamma) \sum_{k \le 2n -1} h^{(\ha - \rho) k} \\
    &\hspace{2cm} + Ct h^{2 - 3 \tilde \rho}(1 + h^{1 - \tilde \rho}\gamma) \sum_{k \le 2n -1} h^{(\ha - \tilde \rho)k} \\
    &\le C\big[h^{2 - 3 \rho} + h^{3 - 4 \rho} \gamma + t (h^{\ha} \gamma^{-\frac{3}{2}} + h \gamma^{-1}) \big], \nonumber
\end{aligned}
\end{equation}
for $t \ge 1$ where we used Lemma~\eqref{lem:trace} to estimate the trace. For $0 \le t \le 1$, we simply have 
\begin{equation}
    \|A(t) - \Op_h^\w(a(t)) \|_{\tr} \le \int_0^t \|\Op_h^\w(e_1(s))\|_{\tr}\, ds \le C t h^{2 - 3 \rho}(1 + h^{1 - \rho}\gamma),
\end{equation}
which gives the first case in~\eqref{eq:tr_est}
\end{proof}
\begin{Remarks}
    1. Theorem~\ref{thm:gaussian_case} follows immediately from Theorem~\ref{thm:general} upon setting $\rho = \ha$ since the standard coherent states belong to symbol class $S^{L^1}_{1/2}$. The theorem also applies to mixtures of not-too-squeezed pure Gaussian states described in Remark 3 following Theorem~\ref{thm:gaussian_case}. 

    \noindent
    2. In the case $\gamma = 1$, we see that we can take $0 \le \rho \le \frac{2}{3}$ and \eqref{eq:L_expansion} would still make sense. The proof of Theorem~\ref{thm:general} works the same for such $\rho$, but due to the sum in \eqref{eq:split}, the bound in trace norm would be more complicated and depend on the dimension. We exclude this case from our presentation since $0 \le \rho \le \ha$ already contains the range of symbols that are mixtures of coherent states. 
\end{Remarks}

\subsection{Higher order correspondence}
In fact, we can construct classical observables that better approximate the quantum evolution by making higher order corrections. This is done by solving the Lindblad evolution asymptotically. 
\begin{theorem}\label{thm:asymptotic}
    Suppose $\mathcal L$ given in~\eqref{eq:full_lindbladian} satisfies~\eqref{eq:PL_assumptions} and~\eqref{eq:jump_elliptic}. Further assume that $h^{2 \rho - 1} \le \gamma \le h^{-1}$. If $A(t)$ satisfies 
    \begin{equation*}
        \partial_t A(t) = \mathcal L A(t), \quad A(0) = \Op_h^\w(a_0), \quad a_0 \in S^{L^1}_\rho,
    \end{equation*}
    then for each $N > 0$, there exists $a_N(t)$ such that 
    \begin{equation}
        \|A(t) - \Op_h^\w(a_N(t))\|_{\tr} \le t^{N} (h^{2 - 3 \rho}(1 + h^{1 - \rho} \gamma))^N.
    \end{equation}
\end{theorem}
\begin{proof}
    We claim that for each $N \ge 0$, there exits $a_N$ and $e_N$ such that $a_N(0) = a_0$, 
    \begin{equation}\label{eq:N}
        \partial_t A_{N} = \mathcal L A_{N} + \Op_h^\w(e_N(t)), \qquad A_N(t) := \Op_h^\w(a_N(t))
    \end{equation}
    and 
    \begin{equation}
        e_N \in t^{N - 1} (h^{2 - 3 \rho}(1 + h^{1 - \rho} \gamma))^N S^{L^1}_\rho.
    \end{equation}
    We proceed by induction. Note that $a_1(t) := e^{tQ} a_0$ achieves the base case with $e_1$ as defined in~\eqref{eq:e1}, which satisfies~\eqref{eq:e1_est}. Now suppose~\eqref{eq:N} holds for some $N > 0$. Then define
    \begin{equation*}
        a_{N + 1} := -\int_0^t e^{(t - s)Q} e_N(s)\, ds. 
    \end{equation*}
    Then it follows that 
    \begin{equation*}
        a_{N + 1} \in t^{N - 1} (h^{2 - 3 \rho}(1 + h^{1 - \rho} \gamma))^N S^{L^1}_\rho, 
    \end{equation*}
    and 
    \begin{equation*}
        \partial_t A_{N + 1} = \mathcal L A_{N + 1} + \Op_h^\w(e_{N + 1}(t))
    \end{equation*}
    for some 
    \begin{equation*}
        e_{N + 1} \in t^{N} (h^{2 - 3 \rho}(1 + h^{1 - \rho} \gamma))^{N + 1} S^{L^1}_\rho,
    \end{equation*}
    which completes the induction. The trace estimate then follows from~\eqref{eq:N} by
    \begin{align*}
        \|A(t) - A_N(t) \|_{\tr} &\le \int_0^t \|e^{(t - s) \mathcal L} \Op_h^\w(e_N(s))\|_{\tr}\, ds \\
        &\le \int_0^t \|\Op_h^\w(e_N(s))\|_{\tr}\, ds \\
        &\le t^{N} (h^{2 - 3 \rho}(1 + h^{1 - \rho} \gamma))^N,
    \end{align*}
    where we again used \cite{Davies77} to see that $e^{t \mathcal L}$ is a contraction on the space of trace class operators. 
\end{proof}
Note that Theorem~\ref{thm:asymptotic} can be slightly improved when we have strict inequality $\gamma > h^{2 \rho - 1}$ for large times by using the smoothing inequality of Proposition~\ref{prop:smoothing} rather than the semiclassical estimate in Corollary~\ref{cor:semiclassical}. However, this does not make a difference in the case $\rho = \ha$ and $\gamma = \mathrm{constant}$, so we only present the theorem in this form for simplicity.

\section*{Acknowledgments}
I would like to thank Felipe Hern\'andez for introducing me to the problem and for many helpful discussions about the methods in \cite{HRRb, HRRa}. I am also grateful to Maciej Zworski for his generous guidance and encouragment on this project as well as his insight to \cite{GZ24}. I would also like to thank Semyon Dyatlov for his support and many helpful suggestions on the project, and Daniel Ranard and Jess Riedel for helpful comments on early drafts of the paper. Partially support from Semyon Dyatlov's DMS-1749858 and DMS-2400090 is acknowledged.


\bibliographystyle{alpha}
\bibliography{FPL.bib}

\begin{thebibliography}{HRR23b}

\bibitem[CP17]{history}
D.~Chru{\'s}ci{\'n}ski and S.~Pascazio.
\newblock A brief history of the {GKLS} equation.
\newblock {\em Open Systems \& Information Dynamics}, 24(03):1740001, 2017.

\bibitem[Dav77]{Davies77}
E.B. Davies.
\newblock Quantum dynamical semigroups and the neutron diffusion equation.
\newblock {\em Reports on Mathematical Physics}, 11(2):169--188, 1977.

\bibitem[DS99]{DS99}
M.~Dimassi and J.~Sjostrand.
\newblock {\em Spectral Asymptotics in the Semi-Classical Limit}.
\newblock London Mathematical Society Lecture Note Series. Cambridge University Press, 1999.

\bibitem[Ejd94]{Ejdelman94}
S.~D. Ejdel'man.
\newblock {\em Parabolic Equations}, pages 203--316.
\newblock Springer Berlin Heidelberg, Berlin, Heidelberg, 1994.

\bibitem[Fri08]{Friedman_parabolic}
A.~Friedman.
\newblock {\em Partial Differential Equations of Parabolic Type}.
\newblock Dover Books on Mathematics. Dover Publications, 2008.

\bibitem[GKS76]{GKS76}
V.~Gorini, A.~Kossakowski, and E.~C.~G. Sudarshan.
\newblock {Completely positive dynamical semigroups of {N}-level systems}.
\newblock {\em Journal of Mathematical Physics}, 17(5):821--825, 05 1976.

\bibitem[GZ24]{GZ24}
J.~Galokowski and M.~Zworski.
\newblock Classical quantum correspondence in {L}indblad evolution, 2024.
\newblock Preprint; \hyperlink{https://arxiv.org/abs/2403.09345}{arXiv:2403.09345}, with an appendix by Huang, Z. and Zworski, M.

\bibitem[HRR23a]{HRRb}
F.~Hern{\'a}ndez, D.~Ranard, and J.~Riedel.
\newblock Decoherence ensures classicality beyond the {Ehrenfest} time as {$\hbar \to 0$}, 2023.
\newblock Preprint; \hyperlink{https://arxiv.org/abs/2306.13717}{arXiv:2306.13717}.

\bibitem[HRR23b]{HRRa}
F.~Hern{\'a}ndez, D.~Ranard, and J.~Riedel.
\newblock The limit of open quantum systems with general {Lindbladians}: vanishing noise ensures classicality beyond the ehrenfest time, 2023.
\newblock Preprint; \hyperlink{https://arxiv.org/pdf/2307.05326}{arXiv:2307.05326}.

\bibitem[Lev07]{Levi07}
E.~E. Levi.
\newblock Sulle equazioni lineari totalmente ellittiche alle derivate parziali.
\newblock {\em Rend. Circ. Matem. Palermo}, 24:275–317, 1907.

\bibitem[Lin76]{Lindblad76}
G.~Lindblad.
\newblock On the generators of quantum dynamical semigroups.
\newblock {\em Commun.Math. Phys.}, 48:119--130, 1976.

\bibitem[Zhi59]{Zhito}
Ya.~I. Zhitomirskij.
\newblock The {C}auchy problem for parabolic systems of linear partial differential equations with growing coefficients.
\newblock {\em Izv. Vyssh. Uchebn. Zaved. Mat.}, 1:55--74, 1959.

\bibitem[Zwo12]{Zworski_semiclassical_analysis}
M.~Zworski.
\newblock {\em Semiclassical analysis}, volume 138 of {\em Graduate Studies in Mathematics}.
\newblock AMS, 2012.

\end{thebibliography}

\end{document}